\documentclass[11pt,a4paper]{article}
\usepackage[T1]{fontenc}
\usepackage{geometry}
\usepackage{amsmath}
\usepackage{amssymb}
\usepackage{amsthm}
\usepackage{mdframed}
\usepackage{graphicx}
\usepackage{bm}
\usepackage[inline]{enumitem}
\usepackage{tikz}
\usetikzlibrary{matrix}
\usepackage{dsfont}
\usepackage[numbers]{natbib}
\usepackage[english]{babel}
\usepackage{caption}
\usepackage[toc,page]{appendix}
\numberwithin{equation}{section}

\newcommand*{\lz}{\lambda}
\newcommand*{\id}{\mathop{}\!\mathrm{d}}
\newcommand*{\js}{\psi}
\newcommand*{\mjs}{\widehat{\js}}

\newcommand*{\josts}{\psi}

\newcommand*{\jss}{\varphi}

\newcommand*{\st}{\sigma_3}

\newcommand{\determz}[3]{#1_{#2}^{\dagger} {#3} #1_{#2} {#3}}

\newcommand*{\C}{\mathbb{C}}

\newcommand*{\R}{\mathbb{R}}
\newcommand*{\I}{\mathds{1}}

\DeclareMathOperator*{\Res}{Res}

\DeclareMathOperator{\sech}{sech}
\DeclareMathOperator{\diag}{diag}
\DeclareMathOperator{\tr}{tr}
\newtheorem{remark}{Remark}

\newtheorem{prop}{Proposition}
\newtheorem{lem}{Lemma}
\newtheorem{theorem}{Theorem}

\newmdtheoremenv{rhp}{Riemann--Hilbert problem}
\newcommand\Item[1][]{%
  \ifx\relax#1\relax  \item \else \item[#1] \fi
  \abovedisplayskip=0pt\abovedisplayshortskip=0pt~\vspace*{-\baselineskip}}
\begin{document}
\title{Dressing a new integrable boundary of the nonlinear
Schr\"odinger equation}
\author{K.T. Gruner%
\thanks{The author is partially supported by the SFB/TRR 191 `Symplectic Structures in Geometry,
Algebra and Dynamics', funded by the DFG.}\\\texttt{kgruner@math.uni-koeln.de}
\\\\ \emph{Universit{\"a}t zu K{\"o}ln, Mathematisches Institut,}
\\ \emph{Weyertal 86-90, 50931 K{\"o}ln, Germany}
}
\date{August 07, 2020}
\maketitle
  \begin{abstract}
  We further develop the method of dressing the boundary for the focusing nonlinear Schr\"odinger
  equation (NLS) on the half-line to include the new boundary condition presented by Zambon.
  Additionally, the foundation to compare the solutions to the ones produced by the mirror-image technique
  is laid by explicitly computing the change of scattering data under the Darboux transformation.
  In particular, the developed method is applied to insert pure soliton solutions.
  \end{abstract}
\section{Introduction}
The nonlinear Schr{\"o}dinger equation is one of the well-known examples, where the model of a physical
phenomenon incorporates both nonlinearity and dispersion in such a way that a soliton---(i) a wave of
permanent form (ii) which is localized (iii) and can interact strongly with other solitons and
retain its identity---emerges. For integrable models as the NLS equation is, these special solutions
have been worked out extensively and primarily by the inverse scattering method. In this context
it should be noted that many physical phenomena naturally arise as initial-boundary value problems,
due to the localized character of the problem. Nonetheless, characterizing soliton solutions for
these problems is substantially less developed than for the corresponding initial value problems,
which is therefore an objective of this work.

Similar to the model of the initial value problem, the integrability is a crucial property to be
able to derive soliton solutions. Besides the usually addressed boundary conditions,
the Dirichlet, Neumann and Robin boundary condition \cite{BiHw}, in that category for the NLS equation,
a new integrable boundary has been derived by dressing the Dirichlet boundary with a defect \cite{Za}.
To then find soliton solutions in these models different method have been successfully applied.

One of these methods, which is called ``dressing the boundary'', utilizes the Darboux transformation in a way
which preserves the integrable structures of the system at the boundary. Again, the usual boundary conditions
in the case of the NLS equation provided a suitable framework to apply this method \cite{Zh}, after it had
already been used to produce results in the case of the sine-Gordon equation on the half-line together
with integrable boundary conditions \cite{ZCZ}. Most importantly,
for these boundary conditions, it has been established that solitons which travel with a specific velocity need
to have a counterpart, we call mirror soliton, with equal amplitude and opposite velocity. This pair of solitons
stand for the reflection happening at the boundary, where the soliton interchanges its role with the
corresponding mirror soliton.

Every soliton has a specific set of so called scattering data in the context of the inverse scattering method,
from which it can be described uniquely. In \cite{BiHw}, apart from using a different method, the mirror-image
technique, to construct soliton solutions in the aforementioned usually addressed models with integrable boundary
conditions, relations regarding the scattering data of the soliton and of the corresponding mirror soliton have
been established. Having these relations facilitates the analysis and understanding of the soliton behavior.

The paper is organized as follows: In Section~\ref{sc:NLS}, we review the inverse scattering method
for the NLS equation in order to derive the expression of a one-soliton solution and its dependency
on the scattering data. Then, we adapt the method of dressing the boundary to the new boundary conditions
in Section~\ref{sc:DTB}. Using these results, we take pure soliton solutions in Section~\ref{sc:rel},
which can be constructed in this theory, and point out common relations between the scattering data
of this solution. Finally, we visualize said solutions in Section~\ref{sc:sol}. 

\section{Initial value problem for the NLS}\label{sc:NLS}
In the following, we give a brief summary of the inverse scattering transform of the
focusing NLS equation. As in \cite{BiHw} and \cite{Fu}, it will serve as a guideline
in order to implement additional results. Therefore, following the analysis given in
\cite{A}, we introduce the NLS equation
\begin{align}\label{eq:NLS}
\begin{split}
  &i u_t + u_{xx} + 2 |u|^2u = 0, \\
  &u(0,x) = u_0(x)
\end{split}
\end{align}
for \(u(t,x)\colon \R \times \R \mapsto \C\) and the initial condition \(u_0(x)\).
The equation can be expressed in an equivalent compatibility condition of the following
linear spectral problems
\begin{align}\label{eq:jost}
\begin{split}
  &\js_x = U\js,\\
  &\js_t = V\js,
\end{split}
\end{align}
where \(\js(t,x,\lz)\) and the matrix operators
\begin{equation}\label{eq:UV}
  U= -i \lz \st + Q, \quad
  V=-2 i\lz^2 \st +\widetilde{Q}
\end{equation}
are \(2\times2\) matrices. The potentials \(Q\) and \(\widetilde{Q}\)
of \(U\) and \(V\) are defined by
\begin{equation*}
  Q(t,x)=
  \begin{pmatrix}
    0 & u \\
    -u^\ast & 0
  \end{pmatrix},\quad
  \widetilde{Q}(t,x,\lz)=
  \begin{pmatrix}
    i|u|^2 & 2\lz u+iu_x \\
    -2\lz u^\ast+iu_x^\ast & -i|u|^2
  \end{pmatrix}\quad \text{and }
  \st=
  \begin{pmatrix}
    1 & 0 \\
    0 & -1
  \end{pmatrix}.
\end{equation*}
In this context, the matrices \(U\) and \(V\) form a so-called Lax pair,
depending not only on \(t\) and \(x\), but also on a spectral parameter \(\lz\).
Hereafter, the asterix denotes the complex conjugate, \(\C_+=\{\lz\in \C \colon
\Im(\lz)>0\}\) as well as \(\C_-=\{\lz\in \C \colon \Im(\lz)<0\}\)
and \(\js^\intercal\) is the transpose of \(\js\).
For a solution \(\js(t,x,\lz)\) of the Lax system~\eqref{eq:jost} the compatibility
condition \(\js_{tx} = \js_{xt}\) for all \(\lambda \in \C\)
is equivalent to \(u(t,x)\) satisfying the NLS equation \eqref{eq:NLS}. Moreover,
we will refer to \(U\) and \(V\) as the \(x\) and \(t\) part of the Lax pair,
respectively.
\subsection{Inverse scattering method for the NLS equation}
In that regard, given a sufficiently fast decaying function
\(u(t,x)\to0\) and derivative \(u_x(t,x)\to0\) as \(|x|\to\infty\), it is reasonable
to assume that there exist \(2\times 2\)-matrix-valued solutions, we call modified Jost
solutions under time evolution, \(\mjs(t,x,\lz)=\js(t,x,\lz) e^{i\theta(t,x,\lz)\st}\),
where \(\theta(t,x,\lz)=\lz x+ 2 \lz^2 t\), of the modified Lax system
\begin{equation*}
  \mjs_x + i \lambda [\st,\mjs] = Q\mjs,\qquad
  \mjs_t + 2i\lambda^2[\st,\mjs]= \widetilde{Q}\mjs
\end{equation*}
with constant limits as \(x\to\pm\infty\) and for all \(\lz\in\R\),
\begin{equation*}
  \mjs_\pm(t,x,\lz) \to \I,  \qquad \text{as }x\to\pm\infty.
\end{equation*}
They are solutions to the following Volterra integral equations:
\begin{equation}\label{eq:Vie}
\begin{aligned}
  \mjs_-(t,x,\lz) &= \I + \int_{-\infty}^{x}
  e^{-i\theta(0,x-y,\lz)\st} Q(t,y) \mjs_-(t,y,\lz) e^{i\theta(0,x-y,\lz)\st}  \id y, \\
  \mjs_+(t,x,\lz) &= \I - \int_{x}^{\infty}
  e^{-i\theta(0,x-y,\lz)\st} Q(t,y) \mjs_+(t,y,\lz) e^{i\theta(0,x-y,\lz)\st}  \id y.
\end{aligned}
\end{equation}
\begin{lem}
Let \(u(t,\cdot)\in H^{1,1}(\R)=\{f\in L^2(\R)\colon xf,f_x\in L^2(\R)\}\).
Then, for every \(\lz\in\R\), there exist unique solutions \(\mjs_\pm(t,\cdot,\lz)
\in L^\infty(\R)\) satisfying the integral equations \eqref{eq:Vie}. Thereby,
the second column vector of \(\mjs_-(t,x,\lz)\) and the first column vector
of \(\mjs_+(t,x,\lz)\) can be continued analytically in \(\lz\in \C_-\) and
continuously in \(\lz\in \C_-\cup\R\), while the first column vector
of \(\mjs_-(t,x,\lz)\) and the second column vector of \(\mjs_+(t,x,\lz)\)
can be continued analytically in \(\lz\in \C_+\) and continuously in
\(\lz\in \C_+\cup\R\).
\end{lem}
Analogously, the columns of \(\js_\pm(t,x,\lz)\) can be continued analytically and
continuously into the complex \(\lz\)-plane, \(\js_-^{(2)}\) and \(\js_+^{(1)}\)
can be continued analytically in \(\lz\in \C_-\) and continuously in \(\lz\in \C_-\cup\R\),
while \(\js_-^{(1)}\) and \(\js_+^{(2)}\) can be continued analytically in
\(\lz\in \C_+\) and continuously in \(\lz\in \C_+\cup\R\).

The limits of the Jost solutions and the zero trace of the matrix \(U\) gives
\(\det \js_\pm = 1\) for all \(x\in\R\). Further, \( \js_\pm\) are both fundamental
matrix solutions to the Lax system \eqref{eq:jost}, so there exists an \(x\)- and
\(t\)-independent matrix \(A(\lz)\) such that
\begin{equation*}
   \js_-(t,x,\lz)=\js_+(t,x,\lz)A(\lz),\qquad \lz\in\R.
\end{equation*}
The scattering matrix \(A\) is determined by this system and therefore we can also
write \(A(\lz) = (\js_+(t,x,\lz))^{-1}\js_-(t,x,\lz)\), whereas its
entries can be written in terms of Wronskians. In particular,
\(a_{11}(\lz) = \det[\js_-^{(1)}|\js_+^{(2)}]\) and \(a_{22}(\lz) =
-\det[\js_-^{(2)}|\js_+^{(1)}]\), implying that they can respectively be
continued in \(\lz\in\C_+\) and \(\lz\in\C_-\). The eigenfunction inherit the
symmetry relation of the Lax pair
\begin{equation}\label{eq:jxsym}
  \js_\pm(t,x,\lz)= -\sigma \bigl(\js_\pm(t,x,\lz^\ast)\bigr)^\ast \sigma,
\end{equation}
which directly gives \(a_{22}(\lz) = a_{11}^\ast(\lz^\ast)\) and
\(a_{21}(\lz)=-a_{12}^\ast(\lz)\). The asymptotic behavior of the modified
Jost functions and scattering matrix as \(\lz\to\infty\) is
\begin{align*}
  \mjs_- & =\I+\frac{1}{2i\lz}\st Q+ \frac{1}{2i\lz}\st
  \int_{-\infty}^{x} |u(t,y)|^2 \id y+\mathcal{O}(1/\lz^2), \\
  \mjs_+ & =\I+\frac{1}{2i\lz}\st Q- \frac{1}{2i\lz}\st
  \int_{x}^{\infty} |u(t,y)|^2 \id y+\mathcal{O}(1/\lz^2),
\end{align*}
and \(A(\lz) = \I + \mathcal{O}(1/k)\).

Let \(u(t,\cdot)\in H^{1,1}(\R)\) be generic.
That is, \(a_{11}(\lz)\) is nonzero in \(\overline{\C_+}\) except at a finite
number of points \(\lz_1,\dots,\lz_N \in\C_+\), where it has simple zeros
\(a_{11}(\lz_j)=0\), \(a'_{11}(\lz_j)\neq0\), \(j=1,\dots,N\).
This set of generic functions \(u(t,\cdot)\) is an open dense subset of
\(H^{1,1}(\R)\) usually denoted by \(\mathcal{G}\). By the symmetry mentioned above,
\(a_{11}(\lz_j)=0\) if and only if \(a_{22}(\lz_j^\ast)=0\)
for all \(j=1,\dots, N\). At these zeros of \(a_{11}\) and \(a_{22}\),
we obtain for the Wronskians the following relation for \(j=1,\dots, N\),
\begin{equation}\label{eq:bj}
  \js_-^{(1)}(t,x,\lz_j)=b_j\js_+^{(2)}(t,x,\lz_j),\quad
  \js_-^{(2)}(t,x,\bar{\lz}_j)=\bar{b}_j\js_+^{(1)}(t,x,\bar{\lz}_j),
\end{equation}
where we defined \(\bar{\lz}_j=\lz_j^\ast\). Whereas for \(j=1,\dots,N\),
the relations then provide residue relations used in the inverse scattering method
\begin{equation*}
\begin{aligned}
  \Res_{\lz=\lz_j} \Bigl(\frac{\mjs_-^{(1)}}{a_{11}}\Bigr)
  &= C_j e^{2i\theta(t,x,\lz_j)}\mjs_+^{(2)}(t,x,\lz_j), \\
  \Res_{\lz=\bar{\lz}_j} \Bigl(\frac{\mjs_-^{(2)}}{a_{22}}\Bigr)
  &= \bar{C}_j e^{-2i\theta(t,x,\bar{\lz}_j)}\mjs_+^{(1)}(t,x,\bar{\lz}_j),
\end{aligned}
\end{equation*}
where the weights are \(C_j=b_j\bigl(\frac{\id a_{11}(\lz_j)}{\id \lz}\bigr)^{-1}\) and \(\bar{C}_j=
\bar{b}_j\bigl(\frac{\id a'_{22}(\bar{\lz}_j)}{\id \lz}\bigr)^{-1}\), and they satisfy the symmetry relations
\(\bar{b}_j=-b_j^\ast\) and \(\bar{C}_j = -C_j^\ast\).

The inverse problem can be formulated using the jump matrix
\begin{equation*}
  J(t,x,\lz) =
  \begin{pmatrix}
    |\rho(\lz)|^2 & e^{-2i\theta(t,x,\lz)}\rho^\ast(\lz) \\
    e^{2i\theta(t,x,\lz)}\rho(\lz) & 0
  \end{pmatrix},
\end{equation*}
where the reflection coefficient is \(\rho(\lz) = a_{12}(\lz)/a_{11}(\lz)\)
for \(\lz\in\R\). Defining sectionally meromorphic functions
\begin{equation*}
  M_-=(\mjs_+^{(1)},\mjs_-^{(2)}/a_{22}),\qquad
  M_+=(\mjs_-^{(1)}/a_{11},\mjs_+^{(2)}),
\end{equation*}
we can give the method of recovering the solution \(u(t,x)\) from the scattering data.
\begin{rhp}\label{rhp}
For given scattering data \((\rho,\{\lz_j,C_j\}_{j=1}^N)\) as well as \(t,x\in \R\),
find a \(2\times2\)-matrix-valued function \(\C\setminus\R\ni\lz\mapsto
M(t,x,\lz)\) satisfying
\begin{enumerate}
  \item \(M(t,x,\cdot)\) is meromorphic in \(\C\setminus\R\).
  \item \(M(t,x,\lz) = 1 + \mathcal{O}(1/\lz)\) as \(|\lz|\to\infty\).
  \item Non-tangential boundary  values \(M_\pm(t,x,\lz)\) exist, satisfying the
  jump condition \(M_+(t,x,\lz)=M_-(t,x,\lz)(1+J(t,x,\lz))\) for
  \(\lz\in\R\).
  \item \(M(t,x,\lz)\) has simple poles at \(\lz_1,\dots,\lz_N,
      \bar{\lz}_1,\dots,\bar{\lz}_N\) with
  \begin{align*}
  \Res_{\lz=\lz_j} M(t,x,\lz) &=\lim_{\lz\to\lz_j} M(t,x,\lz)
  \begin{pmatrix}
    0 & 0 \\
    C_j e^{2i\theta(t,x,\lz_j)} & 0
  \end{pmatrix}, \\
  \Res_{\lz=\bar{\lz}_j} M(t,x,\lz) &=\lim_{\lz\to\bar{\lz}_j}
  M(t,x,\lz)\begin{pmatrix}
    0 & \bar{C}_j e^{-2i\theta(t,x,\bar{\lz}_j)} \\
    0 & 0
  \end{pmatrix}.
  \end{align*}
\end{enumerate}
\end{rhp}
After regularization, the Riemann--Hilbert problem \ref{rhp} can be solved via
Cauchy projectors, and the asymptotic behavior of \(M_\pm(t,x,\lz)\) as \(\lz\to \infty\)
yields the reconstruction formula
\begin{align*}
  u(t,x) & = -2i\sum_{j=1}^{N} C_j^\ast e^{-2i\theta(t,x,\lz_j^\ast)}
  [\mjs_+^\ast]_{22}(t,x,\lz_j)\\
   & \quad - \frac{1}{\pi} \int_{-\infty}^{\infty}  e^{-2i\theta(t,x,\lz)}
   \rho^\ast(\lz) [\mjs_+^\ast]_{22}(t,x,\lz) \id\lz.
\end{align*}
In the reflectionless case, we have \(\rho(\lz)=0\) for \(\lz\in\R\) and
the one-soliton solution with \(\lz_1 = \xi+i\eta\) can be calculated as
\begin{equation*}
  u(t,x)=-2i\eta\frac{C_1^\ast}{|C_1|} e^{-i(2\xi x+4(\xi^2-\eta^2)t)}
  \sech\Bigl(2\eta(x+4\xi t)-\log\frac{|C_1|}{2\eta}\Bigr).
\end{equation*}
We change the notation so that \(u(t,x) = u_{1s}(t,x;\xi,\eta,x_1,\phi_1)\) has the
following expression
\begin{equation}\label{eq:1sol}
  u_{1s}(t,x;\xi,\eta,x_1,\phi_1)=2\eta
  e^{-i(2\xi x+4(\xi^2-\eta^2)t + (\phi_1+\pi/2))}\sech(2\eta(x+4\xi t-x_1)),
\end{equation}
where \(\phi_1= \arg(C_1)\) and \(x_1=\frac{1}{2\eta}\log\frac{|C_1|}{2\eta}\).
\subsection{Dressing the boundary}\label{sc:DTB}
As mentioned in the Introduction, a new integrable boundary condition for
the NLS equation on the half-line has been obtained in
\cite{Za} by dressing a Dirichlet boundary with a ``jump-defect''.
In this section we want to introduce this model on the half-line and
compute soliton solutions via the dressing method.
Therefore, consider the NLS equation \eqref{eq:NLS} for \((t,x)\in
\R_+\times\R_+\) and complement it with boundary condition at \(x=0\),
which, in our notation, are of the following form
\begin{equation}\label{eq:nbc}
  u_x = \frac{i u_t}{2\Omega} - \frac{u \Omega}{2}
  + \frac{u|u|^2}{2\Omega}-\frac{u\alpha^2}{2\Omega},
\end{equation}
where \(\Omega = \sqrt{\beta^2 -|u|^2}\), \(\alpha\) and \(\beta\) real parameters.
Then, the NLS equation has again a corresponding Lax system and the boundary condition can be
written in the form of a boundary constraint
\begin{align}\label{eq:bcv}
  (K_0)_t(t,x,\lz)\big\rvert_{x=0}&=(V(t,x,-\lz) K_0(t,x,\lz) -
  K_0(t,x,\lz) V(t,x,\lz))\big\rvert_{x=0},
\end{align}
where the boundary matrix \(K_0(t,x,\lz)\) is given by
\begin{equation}\label{eq:K0}
  K_0= \frac{1}{(2\lambda - i|\beta|)^2-\alpha^2}
  \begin{pmatrix}
    4 \lambda^2 +4i\lambda \Omega[0]-(\alpha^2+\beta^2) &4i\lambda u[0] \\
    4i\lambda u[0]^\ast&4\lambda^2-4i\lambda\Omega[0]-(\alpha^2+\beta^2)
  \end{pmatrix}.
\end{equation}
The boundary matrix is scaled by \(((2\lambda - i|\beta|)^2-\alpha^2)^{-1}\),
so that
\begin{equation}\label{eq:Kinv}
(K_0(t,x,\lz))^{-1} = K_0(t,x,-\lz).
\end{equation}
Note that the scaling could also be chosen as \(((2\lambda+i|\beta|)^2-\alpha^2)^{-1}\)
leaving the constructed solution unchanged. The property \eqref{eq:Kinv} is needed in the
case of the half-line to properly calculate the zeros and associated kernel vectors of the
special solutions. Moreover, contrary to the boundary constraint on two half-lines, the
boundary constraint on one half-line has a limitation to the \(t\) part of the Lax pair,
as already mentioned in \cite{G}. Nevertheless, it is possible to compute soliton solutions
in this model. Therefore, we introduce the relation of the boundary matrix \(K_0(t,x,\lz)\)
to defect matrices, which are linear in \(\lz\), see \cite{G}.
\begin{prop}
The boundary matrix \(K_0(t,x,\lz)\) can be viewed, up to a function of \(\lambda\),
as product of two defect matrices
\begin{equation*}
2\lambda G_{0,\alpha}(t,x,\lz)=2\lambda\I +
  \begin{pmatrix}
    \alpha\pm i\sqrt{\beta^2-|u|^2} & i u \\
    i u^\ast & \alpha\mp i\sqrt{\beta^2-|u|^2}
  \end{pmatrix},
\end{equation*}
where \(\alpha,\beta\in\R\setminus\{0\}\) and \(\tilde{u}\) is subject to the
Dirichlet boundary condition. In fact,
\begin{equation*}
  ((2\lambda - i\beta)^2-\alpha^2) K_0(t,x,\lz) = 4\lambda^2 G_{0,\alpha}(t,x,\lz) G_{0,-\alpha}(t,x,\lz).
\end{equation*}
\end{prop}
In particular, it is of importance that the product \(K_0(t,x,\lz)\) of the two defect
matrices \(G_{0,\alpha}(t,x,\lz)\) and \(G_{0,-\alpha}(t,x,\lz)\) is commutative.
Thereby, it is comprehensible that a kernel vector for each of the matrices
\(G_{0,\alpha}(t,x,\lz)\) and \(G_{0,-\alpha}(t,x,\lz)\) at particular, different
\(\lz_1, \lz_2\) introduce the same kernel vectors for \(K_0(t,x,\lz)\) at these
values of \(\lz\). In this approach we will leave out boundary-bound soliton solutions.
This is due to the fact that the number of zeros
and associated kernel vectors given through the special solutions is halved
when working with boundary-bound soliton solutions on one half-line.
Referring to the analysis in \cite{G}, we also introduce the function space
\(X = \{f\in H^{1,1}_t(\R_+), \partial_x f \in H^{0,1}_t(\R_+)\}\), where \(H^{0,1}_t(\R_+) =
\{f\in L^2(\R_+)\colon t f\in L^2(\R_+)\}\) and \(H^{1,1}_t(\R_+) = \{f\in L^2(\R_+)\colon
\partial_t f, t f\in L^2(\R_+)\}\). As in the reference, this function space is essential
when it comes to identifying the exact signs of the entries in the diagonal of the constructed
boundary matrix corresponding to the dressed solution.
\begin{prop}\label{p:Nmirsol}
Consider a solution \(u[0](t,x)\) to the NLS equation on the half-line subject to
the new boundary conditions \eqref{eq:bcv} with parameters \(\alpha, \beta
\in \R \setminus\{0\}\) and at \(x=0\) in the function space \(u[0](\cdot,0)\in X\).
Take two solutions \(\{\josts_0, \widehat{\josts}_0\}\)
of the undressed Lax system corresponding to \(u[0]\) for \(\lambda=\lambda_0=
-\frac{\alpha+i\beta}{2}\) and \(\lambda=\widehat{\lambda}_0= -\lambda_0\).
Further, take \(N\) solutions \(\josts_j\) of the undressed Lax system
corresponding to \(u[0]\) for distinct \(\lambda= \lambda_j\in \C \setminus
\bigl(\R\cup i\R\cup\{\lambda_0,\lambda_0^\ast,-\lambda_0,-\lambda_0^\ast\}\bigr)\),
\(j=1,\dots,N\). Constructing \(K_0(t,x,\lz)\) as in \eqref{eq:K0} with \(u[0]\),
\(\alpha\) and \(\beta\), we assume that there exist paired solutions \(\widehat{\josts}_j\)
of the undressed Lax system corresponding to \(u[0]\) for \(\lambda=\widehat{\lambda}_j=
-\lambda_j\), \(j=1,\dots,N\), satisfying
\begin{equation}\label{eq:ssG}
\widehat{\josts}_j\big\rvert_{x=0} = (K_0(t,x,\lz_j) \josts_j)\big\rvert_{x=0},
\quad \widehat{\lz}_k\neq\lz_j.
\end{equation}
Then, a \(2N\)-fold Darboux transformation \(D[2N]\) using
\(\{\josts_1,\widehat{\josts}_1,\dots,\josts_N,\widehat{\josts}_N\}\)
and their respective spectral parameter lead to the solution \(u[2N]\)
to the NLS equation on the half-line. In particular, the boundary condition
is preserved and we denote such a solution \(u[2N]\) by \(\widehat{u}[N]\).
\end{prop}
\begin{proof}
With \(\{\josts_1,\widehat{\josts}_1,\dots,\josts_N, \widehat{\josts}_N\}\),
we have \(2N\) linear independent solutions to the undressed Lax system
\eqref{eq:jost}, since \(\lz_1,\dots,\lz_N,\widehat{\lz}_1,\dots,\widehat{\lz}_N\)
are distinct due to \(\widehat{\lz}_k\neq\lz_j\). Therefore, the Darboux transformation
is uniquely determined and the dressed solution \(\widehat{u}[N]\) satisfies the NLS equation.

In order to proof that there is a matrix \(K_N(t,x,\lz)\) satisfying
\begin{equation*}
  (K_N)_t(t,x,\lz)\big\rvert_{x=0}=(V[2N](t,x,-\lz)K_N(t,x,\lz)-K_N(t,x,\lz) V[2N](t,x,\lz))\big\rvert_{x=0},
\end{equation*}
it is of advantage to consider the equivalent equality
\begin{equation}\label{eq:DK}
  (D[2N](t,x,-\lz) K_0(t,x,\lz))\big\rvert_{x=0} = (K_N(t,x,\lz) D[2N](t,x,\lz))\big\rvert_{x=0},
\end{equation}
where we need to remark that \(K_0(t,x,\lz)\) is multiplied by \(((2\lz-i|\beta|)^2-\alpha^2)/4\)
to simplify further notation. In view of this equality, it becomes plausible to assume that the matrix,
we wish to find, is of second order in \(\lz\), i.e.\@ \(K_N(t,x,\lz)=\lz^2\I + \lz K^{(1)}(t,x)
+K^{(0)}(t,x)\). Our goal will be to construct this matrix \(K_N(t,x,\lz)\) as a Darboux transformation
with spectral parameters \(\lz_0\) and \(-\lz_0\) and corresponding kernel vectors which we need to
determine in the following paragraph. We will restrict the argumentation to one of the spectral
parameters \(\lz_0\) and note that it can be reproduced analogously for the other one \(-\lz_0\).

Since, up to a function of \(\lz\), \(K_0(t,x,\lz)= G_{0,\alpha}(t,x,\lz) G_{0,-\alpha}(t,x,\lz)\),
we can deduce as in \cite{G} that there exist two vectors \(\upsilon_0\) and \(\widehat{\upsilon}_0\)
at two spectral parameters respectively \(\lz_0\) and \(\widehat{\lz}_0= -\lz_0\) for which
\begin{equation*}
  G_{0,\alpha}(t,x,\lz_0)\upsilon_0=0, \quad
  G_{0,-\alpha}(t,x,\widehat{\lz}_0)\widehat{\upsilon}_0=0.
\end{equation*}
Therefore, \(K_0(t,x,\lz)\) can be seen as two-fold dressing matrix with the
inherited kernel vectors of \(G_{0,\alpha}\) and \(G_{0,-\alpha}\) at respectively
\(\lz_0\) and \(\widehat{\lz}_0\), so that
\begin{equation*}
  K_0(t,x,\lz_0)\upsilon_0=0, \quad
  K_0(t,x,\widehat{\lz}_0)\widehat{\upsilon}_0=0.
\end{equation*}
These kernel vectors \(\upsilon_0\) and \(\widehat{\upsilon}_0\) are either linear dependent
or linear independent of \(\{\js_0, \widehat{\js}_0,\jss_0, \widehat{\jss}_0\}\). Further,
these vectors will serve as a means to construct the kernel vectors for the dressing matrix \(K_N(t,x,\lz)\).
Thereby, we distinguish the two cases:
\begin{enumerate}
\item
The kernel vector \(\upsilon_0\) of \(K_0(t,x,\lz_0)\) can be expressed as a linear combination of
\(\{\js_0,\jss_0\}\) at \(\lz=\lz_0\) and \(x=0\). Then, w.l.o.g.\@ \(\upsilon_0 = \js_0\), again to simplify notation.
Since \(\js_0\) is linearly independent of \(\js_1,\dots,\js_N\), it is possible to define a new vector
\begin{equation*}
  \upsilon'_0 = D[2N](t,x,\lz_0) \upsilon_0,
\end{equation*}
which will serve as one of the kernel vectors for the dressing matrix \(K_N(t,x,\lz)\).
It is important to note that constructing \(K_N(t,x,\lz)\) in this manner will result in
the following relations for the vector \(\js_0\) and the orthogonal vector \(\jss_0\) at \(x=0\):
\begin{align}\label{eq:KN1p1}
  \begin{aligned}
  &D[2N](t,x,-\lz_0)K_0(t,x,\lz_0)\js_0=K_N(t,x,\lz_0)D[2N](t,x,\lz_0)\js_0=0,\\
  &D[2N](t,x,-\lz_0^\ast)K_0(t,x,\lz_0^\ast)\jss_0=K_N(t,x,\lz_0^\ast) D[2N](t,x,\lz_0^\ast)\jss_0=0.
\end{aligned}
\end{align}
\item
The kernel vector \(\upsilon_0\) of \(K_0(t,x,\lz_0)\) can not be expressed as a linear combination of
\(\{\js_0,\jss_0\}\) at \(\lz=\lz_0\) and \(x=0\). In this case, making out the kernel vector directly turns out
to be not as easy. Since we assumed that the kernel vector \(\upsilon_0\) and \(\{\js_0,\jss_0\}\) are
linearly independent, we know that at \(x=0\) we can define a vector \(\widetilde{\psi}_0=K_0(t,x,\lz_0)\js_0\neq0\),
which, in particular, solves the \(t\)-part of the undressed Lax system at \(x=0\) and \(\lz=-\lz_0\),
due to \(K_0(t,x,\lz_0)\) satisfying the boundary constraint \eqref{eq:bcv}. Similarly, the relations
\eqref{eq:DN} for the dressing matrix \(D[2N](t,x,\lz)\) imply that \(\js'_0=D[2N](t,x,\lz_0)\js_0\)
and \(\widetilde{\js}'_0=D[2N](t,x,-\lz_0)\widetilde{\js}_0\) are also solutions the \(t\)-part of the
dressed Lax system at \(x=0\), \(\lz=\lz_0\) and \(\lz=-\lz_0\), respectively. Now, connecting these
three transformation, we require that there exists a matrix \(D[2N](t,x,-\lz) K_0(t,x,\lz)
(D[2N](t,x,\lz))^{-1}\), we also call \(K_N(t,x,\lz)\), which then, at \(x=0\), satisfies
\begin{align}\label{eq:KN2p1}
  \begin{aligned}
  &(D[2N](t,x,-\lz_0)K_0(t,x,\lz_0)\js_0)=
  (K_N(t,x,\lz_0)D[2N](t,x,\lz_0)\js_0)\neq0,\\
  &(D[2N](t,x,-\lz_0^\ast)K_0(t,x,\lz_0^\ast)\jss_0)=
  (K_N(t,x,\lz_0^\ast) D[2N](t,x,\lz_0^\ast)\jss_0)\neq0.
\end{aligned}
\end{align}
Further, evaluating the determinant of \(K_N(t,x,\lz)\) at the spectral parameter \(\lz_0\)
or \(\lz_0^\ast\) and \(x=0\), we obtain \(\det(K_N(t,x,\lz_0))=\det(K_N(t,x,\lz_0^\ast))=0\). This implies
that there exist two kernel vectors of \(K_N(t,x,\lz)\) corresponding to one of the spectral
parameters each. Hence, in this case we have found the kernel vector with which we want to
construct the dressing matrix \(K_N(t,x,\lz)\), whereas it then satisfies \eqref{eq:KN2p1}.
\end{enumerate}
Now, given we have constructed a \(2\)-fold Darboux transformation using \(\lz_0\) and \(-\lz_0\)
and the appropriately chosen kernel vectors through the already mentioned procedure, we want to
proof that this dressing matrix \(K_N(t,x,\lz)\) indeed satisfies the equality \eqref{eq:DK}.
First, we will write the equality as matrix polynomials of degree \(2N+2\) in \(\lz\) and
denote them as \(L(\lz)\) and \(R(\lz)\). Hence,
\begin{align*}
&L(\lz)= D[2N](t,0,-\lz) K_0(t,0,\lz) =
\lz^{2N+2} L_{2N+2}+ \lz^{2N+1} L_{2N+1}+ \dots + \lz L_1+ L_0,\\
&R(\lz)= K_N(t,0,\lz) D[2N](t,0,\lz) =
\lz^{2N+2} R_{2N+2}+ \lz^{2N+1} R_{2N+1}+ \dots + \lz R_1+ R_0.
\end{align*}
The structure of the matrices yields \(L_{2N+2} = \I = R_{2N+2}\), which also stems from the
fact that we multiplied \(K_0(t,x,\lz)\) by \(((2\lz-i|\beta|)^2-\alpha^2)/4\). Regarding the
property \eqref{eq:Kinv}, the special solutions provide \(4N\) zeros and associated kernel vectors
\begin{align*}
  &R(\lz)\big\rvert_{\lz=\lz_j}\js_j=0,
  && L(\lz)\big\rvert_{\lz=\lz_j}\js_j=0,\\
  &R(\lz)\big\rvert_{\lz=\widehat{\lz}_j}\widehat{\js}_j=0,
  && L(\lz)\big\rvert_{\lz=\widehat{\lz}_j}\widehat{\js}_j=0,
\intertext{at \(x=0\) for \(j=1,\dots,N\). For \(R(\lz)\) the equalities are clear from the definition
of the Darboux transformation and with the assumption \eqref{eq:ssG}, the equalities for \(L(\lz)\)
follow analogously. With the orthogonal vectors \(\jss_j = \sigma_2 \js_j^\ast\) and \(\widehat{\jss}_j
= \sigma_2 \widehat{\js}_j^\ast\), we obtain}
  &R(\lz)\big\rvert_{\lz=\lz_j^\ast}\jss_j=0,
  && L(\lz)\big\rvert_{\lz=\lz_j^\ast}\jss_j=0,\\
  &R(\lz)\big\rvert_{\lz=\widehat{\lz}_j^\ast}\widehat{\jss}_j=0,
  &&L(\lz)\big\rvert_{\lz=\widehat{\lz}_j^\ast}\widehat{\jss}_j=0,
\end{align*}
whereby these equalities hold at \(x=0\) for \(j=1,\dots,N\). This is however not enough to ensure equality
in \eqref{eq:DK}, since the determinant is of power \(4N+2\) in \(\lz\) and we only have \(4N\) zeros. Further,
the choice of \(K_N\) implies
\begin{align*}
  &R(\lz)\big\rvert_{\lz=\lz_0}\js_0=
  L(\lz)\big\rvert_{\lz=\lz_0}\js_0,\\
  &R(\lz)\big\rvert_{\lz=\widehat{\lz}_0}\widehat{\js}_0=
  L(\lz)\big\rvert_{\lz=\widehat{\lz}_0}\widehat{\js}_0,
\intertext{and with the orthogonal vectors \(\jss_0 = \sigma_2 \js_0^\ast\) and
\(\widehat{\jss}_0 = \sigma_2 \widehat{\js}_0^\ast\), we obtain}
  &R(\lz)\big\rvert_{\lz=\lz_0^\ast}\jss_0=
  L(\lz)\big\rvert_{\lz=\lz_0^\ast}\jss_0,\\
  &R(\lz)\big\rvert_{\lz=\widehat{\lz}_0^\ast}\widehat{\jss}_0=
  L(\lz)\big\rvert_{\lz=\widehat{\lz}_0^\ast}\widehat{\jss}_0,
\end{align*}
at \(x=0\). At this point, it is important that all vectors are linearly independent. In view of
the additional vectors from the construction of \(K_N(t,x,\lz)\), at \(\lz=\lz_0\), we see that in both
cases either a linear combination of \(\js_0\) and \(\jss_0\) or \(\js_0\) itself is given. Therefore,
this provides a linear independent vector and the same is true for the second spectral parameter.
As mentioned in the second case, these vectors are not necessarily kernel vectors for \(L(\lz)\)
and \(R(\lz)\). However, they are indeed kernel vectors of the difference \(C(\lz) = L(\lz)- R(\lz)
= \lz^{2N+1} C_{2N+1} + \dots + C_0\), which can therefore be calculated explicitly with the given
amount of zeros and kernel vectors
\begin{equation*}
  (C_{2N+1},\cdots,C_0)
  \begin{pmatrix}
    \lz_0^{2N+1} \js_0 & \cdots & (\widehat{\lz}_{N}^\ast)^{2N+1}\widehat{\jss}_{N} \\
    \vdots & \vdots & \vdots \\
    \js_0 & \cdots & \widehat{\jss}_{N}
  \end{pmatrix}=0.
\end{equation*}
Consequently, the matrix coefficients of \(C(\lz)\) are zero and so is \(C(\lz)\), implying \(L(\lz) = R(\lz)\).
Furthermore, this equality gives us that in both cases, the kernel vectors are indeed as described in the first
case equal to \(D[2N](t,x,\lz_0)\upsilon_0\) and \(D[2N](t,x,\widehat{\lz}_0)\widehat{\upsilon}_0\) with
\(\upsilon_0\) and \(\widehat{\upsilon}_0\) kernel vectors of \(K_0(t,x,\lz)\) respectively at \(\lz=\lz_0\)
and \(\lz=-\lz_0\). Even though, the linear dependence of the kernel vector to the pair \(\{\js_0,\jss_0\}\)
is not implied in both cases.

Given \(K_N(t,x,\lz)\) of the form \(\lz^2\I + \lz K^{(1)}(t,x)+K^{(0)}(t,x)\), we want to determine
the matrix coefficients in terms of the solution at \(x=0\) to confirm that the boundary conditions
are preserved. Thereby, the symmetry of \(V^\ast(t,x,\lz^\ast)= \sigma V(t,x,\lz) \sigma^{-1}\), where
\(\sigma=
\begin{pmatrix}
  0 & 1 \\
  -1 & 0
\end{pmatrix}\), is inherited by \(K_N(t,x,\lz)\), so that we can identify
\begin{equation*}
  K_N(t,0,\lz)=\lz^2 \I + \lz
  \begin{pmatrix}
    K^{(1)}_{11}(t) & K^{(1)}_{12}(t) \\
    -\bigl(K^{(1)}_{12}(t)\bigr)^\ast & \bigl(K^{(1)}_{11}(t)\bigr)^\ast
  \end{pmatrix} +
  \begin{pmatrix}
    K^{(0)}_{11}(t) & K^{(0)}_{12}(t) \\
    -\bigl(K^{(0)}_{12}(t)\bigr)^\ast & \bigl(K^{(0)}_{11}(t)\bigr)^\ast
  \end{pmatrix}.
\end{equation*}
The equality \(L_{2N+1}=R_{2N+1}\) gives  at \(x=0\) on the off-diagonal of \(K^{(1)}(t,0)\) that
\(K^{(1)}_{12}(t)=i u[2N](t,0)\) and \(K^{(1)}_{21}(t)=-\bigl(K^{(1)}_{12}(t)\bigr)^\ast = iu^\ast[2N](t,0)\).
For the entries on the diagonal of \(K^{(1)}(t,x)\), we obtain from the same equality
\begin{equation}\label{eq:K11}
\begin{aligned}
  K^{(1)}_{11}(t)&=i \Omega[0] - 2 (\Sigma_1)_{11},\\
  \bigl(K^{(1)}_{11}(t)\bigr)^\ast&=-i \Omega[0] - 2 (\Sigma^\ast_1)_{11},
\end{aligned}
\end{equation}
where \(\Sigma_1\) is defined as the matrix coefficient of \(\lz^{2N-1}\) of the matrix \(D[2N](t,x,\lz)\),
see \eqref{eq:S1}. To determine the remaining entries of the matrix coefficients,
we need to extract information from the determinant of \(K_N(t,x,\lz)\), which can be calculated as
\begin{align*}
  \det(K_N(t,0,\lz)) &=\det(D[2N](t,0,-\lz))\det(K_0(t,0,\lz))\det((D[2N](t,0,\lz))^{-1}),
  \intertext{and using that \(\det(D[2N](t,0,-\lz))= \det(D[2N](t,0,\lz))\), due to the fact
  that for each spectral parameter \(\lz_j\) we also use \(-\lz_j\) for \(j=1,\dots,N\), we obtain}
   &= \det(K_0(t,0,\lz))= \lambda^4 - \frac{\alpha^2-\beta^2}{2}\lambda^2 + \frac{(\alpha^2+\beta^2)^2}{16}.
\end{align*}
Formally, calculating the determinant of the matrix \(K_N(t,0,\lz)\) in polynomial form as above, we can
match the coefficients such that
\begin{equation}\label{eq:Kdet}
\begin{aligned}
  \tr(K^{(1)}(t,0)) &=  \vphantom{\frac{1}{1}}0,\\
  \tr(K^{(0)}(t,0))+\det({K^{(1)}}(t,0))&=-\frac{\alpha^2-\beta^2}{2},\\
  2\Re(K^{(1)}_{11}(t)\bigl(K^{(0)}_{11}(t)\bigr)^\ast) -2 \bigl(K^{(0)}_{12}(t)\bigr)^\ast
  \Im(u[2N](t,0))&=\vphantom{\frac{1}{1}}0,\\
  \det(K^{(0)}(t,0))&=\frac{(\alpha^2+\beta^2)^2}{16}.
\end{aligned}
\end{equation}
Combining the first line in \eqref{eq:Kdet} with the expressions we have for \(K^{(1)}_{11}(t)\) and its complex
conjugate, see \eqref{eq:K11}, we can deduce that \(\Re(K^{(1)}_{11}(t))=0\).
Further, evaluating the next equality of \eqref{eq:DK} at \(x=0\), which is \(L_{2N}=R_{2N}\), implies
\begin{equation*}
  L_{2N}=\Sigma_2 - i \Sigma_1
  \begin{pmatrix}
    \Omega[0] & u[0] \\
    u[0]^\ast & -\Omega[0]
  \end{pmatrix}-
  \frac{\alpha^2+\beta^2}{4} \I = \Sigma_2 + K^{(1)}(t,0)\Sigma_1 + K^{(0)}(t,0)=R_{2N},
\end{equation*}
where again \(\Sigma_2\) is the matrix coefficient of \(\lambda^{2N-2}\) of the matrix \(D[2N](t,x,\lz)\).
Matching the \((12)\)-entry of this equality, we derive
\begin{align*}
  (u[2N]-u[0])\frac{\Omega[0]}{2}-i u[0](\Sigma_1)_{11} &=
  -\frac{i}{2}K^{(1)}_{11}(t)(u[2N]-u[0]) + iu[2N](\Sigma^\ast_1)_{11}+K^{(0)}_{12}(t),
\intertext{and using the expressions in \eqref{eq:K11} we have for \((\Sigma_1)_{11}\) and
\((\Sigma^\ast_1)_{11}\), we obtain after cancellation}
  0&=K^{(0)}_{12}(t)-i u[2N] \Re(K^{(1)}_{11}(t)).
\end{align*}
However, we already calculated that \(\Re(K^{(1)}_{11}(t))\) needs to be zero in order for the
determinants to be equal. Hence, also \(K^{(0)}_{12}(t)\) and thereby the off-diagonal of
\(K^{(0)}(t,0)\) vanishes. It follows by the third equation of \eqref{eq:Kdet} that
\(\Im(K^{(0)}_{11}(t))=0\) and then, by the fourth equation we have \(K^{(0)}(t,0)=\pm
\frac{\alpha^2+\beta^2}{4}\I\). To verify that it is indeed minus as for \(K_0(t,0,\lz)\),
we confirm with the equality of \(L_0=R_0\) at \(x=0\), which is
\begin{equation*}
  -\frac{\alpha^2+\beta^2}{4}\Sigma_{2N} = K^{(0)}(t,0) \Sigma_{2N},
\end{equation*}
where \(\Sigma_{2N}\) is the zero-th order matrix coefficient of the dressing matrix \(D[2N](t,x,\lz)\).
For this to be satisfied for all \(t\in\R_+\), we need to have
\(K^{(0)}=-\frac{\alpha^2+\beta^2}{4}\I\). Thereby, we obtain
\(\tr(K^{(0)}(t,0)) = -\frac{\alpha^2+\beta^2}{2}\). Thus, the second
equation of \eqref{eq:Kdet} implies that
\begin{equation*}
  \begin{aligned}
  K^{(1)}_{11}(t) &= \pm i\sqrt{\beta^2-|u[2N]|^2},\\
   \bigl(K^{(1)}_{11}(t)\bigr)^\ast &= \mp i\sqrt{\beta^2-|u[2N]|^2}.
  \end{aligned}
\end{equation*}
Now, we need to determine the sign of the diagonal entries of \(K^{(1)}(t,0)\)
to be able to constitute that \(K_N(t,0,\lz)\) preserves the boundary constraint,
i.e.\@ we need to show that the signs coincide with the signs in the same entry
of \(K_0(t,0,\lz)\) in front of \(\Omega[0]\).

Therefore, a similar analysis as in \cite{G} is needed, where we use the fact that
under the Darboux transformation functions \(u[0](\cdot,0)\) in the function space \(X\)
are mapped onto functions, here \(u[2N](\cdot,0)\),  which lie in the function space \(X\).
Further, we have that \(K_0(t,0,\lz)\) has a positive sign in the \((11)\)-entry in front of \(\Omega[0]\).
As before, we have the kernel vectors \(\upsilon_0\) and \(\widehat{\upsilon}_0\) of \(K_0(t,x,\lz)\)
at \(x=0\) and respectively \(\lz=\lz_0\) and \(\lz=\widehat{\lz}_0\). Then, for \(K_0(t,0,\lz)\) multiplied by
\(((2\lz-i|\beta|)^2-\alpha^2)/4\), we have as \(t\) goes to infinity that
\begin{align*}
  \lim_{t\to\infty} K_0(t,0,\lz) &= \diag(\lz^2+i|\beta|\lz-\frac{(\alpha^2+\beta^2)}{4},
  \lz^2-i|\beta|\lz-\frac{(\alpha^2+\beta^2)}{4})\\
  &=
  \begin{cases}
    \diag((\lz-\lz_0)(\lz-\widehat{\lz}_0^\ast),(\lz-\lz_0^\ast)(\lz-\widehat{\lz}_0)), & \mbox{if } \beta>0,\\
    \diag((\lz-\lz_0^\ast)(\lz-\widehat{\lz}_0),(\lz-\lz_0)(\lz-\widehat{\lz}_0^\ast)), & \mbox{if } \beta<0.
  \end{cases}
\end{align*}
In turn, this implies that the limits of the kernel vectors of \(K_0(t,0,\lz)\) are
\begin{align*}
  &\upsilon_0 \sim
  \begin{cases}
    e_1, & \mbox{if } \beta>0, \\
    e_2, & \mbox{if } \beta<0,
  \end{cases}
  && \widehat{\upsilon}_0 \sim
  \begin{cases}
    e_2, & \mbox{if } \beta>0, \\
    e_1, & \mbox{if } \beta<0,
  \end{cases}
\end{align*}
as \(t\) goes to infinity, where \(e_1\) and \(e_2\) are unit vectors.
Since the dressing matrix \(D[2N](t,x,\lz)\) is also diagonal as \(t\) goes to infinity,
see \cite{G}, the kernel vectors \(\upsilon'_0\), \(\widehat{\upsilon}'_0\) of
\(K_N(t,0,\lz)\) inherit the long time behavior of their corresponding vector.
Therefore, the signs can be determined to be positive in the
\((11)\)-entry and negative in the \((22)\)-entry of \(K^{(1)}\).
In conclusion, if we assume \(u[0](\cdot,0)\in X\), we can find
a Darboux transformation \(K_N(t,x,\lz)\) for which \(V[2N]\)
satisfies \eqref{eq:bcv} regarding \(x=0\), where \(K_N(t,x,\lz)\)
is similar to \(K_0(t,x,\lz)\) with an updated function \(\widehat{u}[N]\).
\end{proof}
\begin{remark}
Similar to the analysis of the long time behavior of the kernel vectors, one could look at the
long time behavior of the dressing matrix \(D[2N](t,x,\lz)\) to deduce the same result through
the equality of \(K_N(t,x,\lz)\) with the product of matrices
\(D[2N](t,x,-\lz) K_0(t,x,\lz) (D[2N](t,x,\lz))^{-1}\). Nevertheless, this is closely related to
one another, since the limit behavior of the kernel vectors of \(D[2N](t,x,\lz)\) determines the
distribution of factors \(\lz-\lz_j\), \(\lz-\widehat{\lz}_j\), \(\lz-\lz_j^\ast\) and
\(\lz-\widehat{\lz}_j^\ast\) for \(j=1,\dots,N\) in the diagonal entries as \(t\) goes to infinity.
\end{remark}

Therewith, we have shown that the method of dressing the boundary can be applied to the new boundary
conditions constituted in \cite{Za}. Unlike in \cite{Zh}, where the boundary matrices \(G_0(\lz)\) and
\(G_N(\lz) = G_0(\lz)\) are provided at the beginning and the proof is to check that they satisfy the
equality \eqref{eq:DK} together with the dressing matrix \(D[2N](t,x,\lz)\), we only provide \(K_0(t,x,\lz)\)
and the proof is to construct a suitable boundary matrix \(K_N(t,x,\lz)\) satisfying the equality \eqref{eq:DK}.
Afterwards, we need to verify that the constructed matrix \(K_N(t,x,\lz)\) is in terms of the solution space
indeed the boundary matrix we need for the boundary constraint with respect to the dressed solution \(\widehat{u}[N]\).
The reason why we need a different approach is due to the boundary matrix. In the case of the Robin boundary condition,
the structure of the boundary matrix \(G_0(\lz)\) is such that when comparing \(L(\lz)\) with \(R(\lz)\) it is already
clear with regard to \(\lz\) that the \((2N+1)\)-th and zero-th order matrix coefficients are equal. Hence,
the zeros and associated kernel vectors of the dressing matrix \(D[2N](t,x,\lz)\) are sufficient to
derive the equality of \(L(\lz)\) and \(R(\lz)\). However defining \(K_N(t,x,\lz)\) similarly to
\(K_0(t,x,\lz)\) as in \eqref{eq:K0} with \(u[0]\) and \(\Omega[0]\) updated to \(u[2N]\) and \(\Omega[2N]\),
we only have the equality of the \((2N+1)\)-th order matrix coefficient and consequently the zeros
and associated kernel vectors of the dressing matrix \(D[2N](t,x,\lz)\) are insufficient to derive
the equality. Even though the new boundary condition does fit in the proof initially suggested
in \cite{Zh}, the Robin boundary condition very well fits into this one as we will put forward
in Appendix \ref{app:rob}.
%
%

\section{Soliton solutions}
\subsection{Relations between scattering data}\label{sc:rel}
In this section, we want to derive relations for multi-soliton solutions between the weights
\(C_j\) and \(\widehat{C}_j\), which correspond to a pair of zeros \(\lz_j\) and \(\widehat{\lz}_j\)
of \(a_{11}(\lz)\), for \(j=1,\dots,N\).

Consider the zero seed solution \(u[0]=0\) and \(\C\setminus (\R\cup i\R)\ni \lz_0 =
-\frac{\alpha + i\beta}{2}\). Also,
\begin{equation*}
  K_0(t,x,\lz) = \frac{1}{(2\lambda - i|\beta|)^2-\alpha^2}
  \begin{pmatrix}
    4 \lambda^2 +4i\lambda |\beta|-(\alpha^2+\beta^2) &0 \\
    0 &4\lambda^2-4i\lambda|\beta|-(\alpha^2+\beta^2)
  \end{pmatrix}.
\end{equation*}
Then, in order to apply Proposition \ref{p:Nmirsol}, we take distinct \(\lz=\lz_j\in\C\setminus
\bigl(\R\cup i\R\cup\{\lz_0,\lz_0^\ast,-\lz_0,-\lz_0^\ast\}\bigr)\), \(j=1,\dots,N\), and at the
corresponding spectral parameter a solution \(\js_j\), \(j=0,\dots,N\), to the Lax system regarding
\(u[0]\). The solutions \(\js_j\) to the zero seed solution are readily produced
\begin{equation*}
  \js_j = \begin{pmatrix}
            \mu_j \\
            \nu_j
          \end{pmatrix}
          = e^{-i(\lz_j x+ 2 \lz_j^2 t)\st}\begin{pmatrix}
            u_j \\
            v_j
          \end{pmatrix},
\end{equation*}
where \((u_j,v_j)^\intercal \in \C^2\) for \(j=0,\dots,N\). Particularly, the choice for \(\js_0\) will
be \((u_0,v_0)= (1,0)\) or \((u_0,v_0)= (0,1)\) respectively for \(\beta>0\) or \(\beta<0\) inspired by
the first case in the proof of Proposition \ref{p:Nmirsol}. Now, given the relation \eqref{eq:ssG}, we also
take solutions to the same Lax system at \(\lz=\widehat{\lz}_j=-\lz_j\) defined by \(\widehat{\js}_j =
e^{-i(\widehat{\lz}_j x + 2 \widehat{\lz}_j^2 t)\st} (\hat{u}_j,\hat{v}_j)^\intercal\), \(j=1,\dots,N\).
Whereas, the relation is equivalent to
\begin{equation*}
  \frac{\hat{u}_j}{\hat{v}_j} = \frac{(2\lz_j+i|\beta|)^2-\alpha^2}{(2\lz_j-i|\beta|)^2-\alpha^2}
  \frac{u_j}{v_j}, \quad j=1,\dots,N.
\end{equation*}
Note that if \(\lz_j\in \C_+\), then \(\widehat{\lz}_j\in \C_-\) which, in turn, implies
that \(\widehat{\js}_j\) has opposite limit behavior as \(\js_j\) for \(x\to\pm\infty\). Such that
in order to apply Theorem \ref{t:scadt} to the Darboux transformation corresponding to
\(\widehat{\lz}_j\) and \(\widehat{\js}_j\), we instead use the counterpart \(\widehat{\lz}_j^\ast\)
and \(\widehat{\jss}_j\). Since with \(\widehat{\lz}_j^\ast\in \C_+\), the vector
\(\widehat{\jss}_j = e^{-i(\widehat{\lz}_j^\ast x+ 2 (\widehat{\lz}_j^\ast)^2 t)\st}
(-\hat{v}_j^\ast,\hat{u}_j^\ast)^\intercal\), admits the same limit behavior as \(\js_j\)
for \(x\to\pm\infty\). Similar to Remark \ref{r:scadt} following Theorem \ref{t:scadt}, we can
deduce for a two-fold Darboux transformation consisting of \(\{\lz_1, \js_1, \widehat{\lz}_1^\ast,
\widehat{\jss}_1\}\) that the weights in the scattering data can be calculated as
\begin{equation*}
  C_1^{(2)} =  -\frac{v_1}{u_1} \frac{(\lz_1-\lz_1^\ast)(\lz_1-\widehat{\lz}_1)}{
  \lz_1-\widehat{\lz}_1^\ast},\qquad
  C_2^{(2)} =  -\frac{\hat{u}_1^\ast}{-\hat{v}_1^\ast} \frac{(\widehat{\lz}_1^\ast-\lz_1^\ast)
  (\widehat{\lz}_1^\ast-\widehat{\lz}_1)}{\widehat{\lz}_1^\ast-\lz_1}.
\end{equation*}
This results in
\begin{equation*}
  C_1^{(2)}(C_2^{(2)})^\ast = -4 \lambda_1^2\cdot \frac{(2\lz_1+i|\beta|)^2-\alpha^2}{(2\lz_1-i|\beta|)^2-\alpha^2} \cdot
  \frac{\Im(\lz_1)^2}{\Re(\lz_1)^2},
\end{equation*}
where it is obvious that the factor \(\frac{(2\lz_1+i|\beta|)^2-\alpha^2}{(2\lz_1-i|\beta|)^2-\alpha^2}\)
is the only difference to the analogous result in the case of the Robin boundary condition,
where we have \(\frac{i\alpha-2\lz_1}{i\alpha+2\lz_1}\). Nevertheless, by defining \(\lz_1 = \xi_1+i\eta_1\)
and \(\widehat{\lz}_1^\ast = -\xi_1+i\eta_1\) as well as the corresponding weights \(C_1=
2 \eta_1 e^{2\eta_1 x_1 + i \phi_1}=C_1^{(2)}\) and \(\widehat{C}_1= 2 \eta_1 e^{2\eta_1 \widehat{x}_1 +i\widehat{\phi}_1}
=C_2^{(2)}\), we obtain a relation between the initial positions and phases of the two inserted solitons
\begin{align*}
  x_1 + \widehat{x}_1 & = \frac{1}{2\eta_1}\log\Bigl(1+\frac{\eta_1^2}{\xi_1^2}\Bigr)+\frac{1}{4\eta_1}
  \log\Bigl(\frac{(4\xi_1^2-\alpha^2-(2\eta_1+|\beta|)^2)^2 +(4\xi_1 (2\eta_1+|\beta|))^2}{
  (4\xi_1^2-\alpha^2-(2\eta_1-|\beta|)^2)^2 +(4\xi_1 (2\eta_1-|\beta|))^2}\Bigr),\\
  \phi_1 - \widehat{\phi}_1 & = 2\arg(\lambda_1) + \arg\Bigl(
  \frac{4\xi_1^2-\alpha^2-(2\eta_1+|\beta|)^2 + i 4\xi_1 (2\eta_1+|\beta|)}{
  4\xi_1^2-\alpha^2-(2\eta_1-|\beta|)^2 + i 4\xi_1 (2\eta_1-|\beta|)}\Bigr)+ \pi.
\end{align*}
\begin{remark}
In general, we can construct a \(2N\)-Darboux transformation using the information given by
\(\{\lz_1, \js_1, \dots, \lz_N, \js_N, \widehat{\lz}_1^\ast, \widehat{\jss}_1,\dots,
\widehat{\lz}_N^\ast, \widehat{\jss}_N\}\), where \(\lz_j = \xi_j + i \eta_j\) and consequently
\(\widehat{\lz}_1^\ast = -\xi_j + i \eta_j\) for \(j=1,\dots,N\) with corresponding solutions to
the undressed Lax system as above. Then for \(j=1,\dots,N\), the relation for a pair of initial
positions \(x_j\) and \(\widehat{x}_j = x_{N+j}\) as well as phases \(\phi_j\) and
\(\widehat{\phi}_j = \phi_{N+j}\) amounts to
\begin{align*}
  x_j + \widehat{x}_j & = \frac{1}{2\eta_j}\log\Bigl(1+\frac{\eta_j^2}{\xi_j^2}\Bigr)+
  \frac{1}{4\eta_j}\log\Bigl(\frac{(4\xi_j^2-\alpha^2-(2\eta_j+|\beta|)^2)^2 +(4\xi_j (2\eta_j+|\beta|))^2}{
  (4\xi_j^2-\alpha^2-(2\eta_j-|\beta|)^2)^2 +(4\xi_j (2\eta_j-|\beta|))^2}\Bigr)\\
  &\quad -\frac{1}{2\eta_j} \sideset{}{'}\sum_{k=1}^{N} \log
  \frac{[(\xi_j-\xi_k)^2+(\eta_j-\eta_k)^2][(\xi_j+\xi_k)^2+(\eta_j-\eta_k)^2]}{
  [(\xi_j+\xi_k)^2+(\eta_j+\eta_k)^2][(\xi_j-\xi_k)^2+(\eta_j+\eta_k)^2]},\\
  \phi_j - \widehat{\phi}_j & = 2\arg(\lz_j) + \arg\Bigl(\frac{4\xi_j^2-\alpha^2-(2\eta_j+|\beta|)^2 +
  i 4\xi_j (2\eta_j+|\beta|)}{4\xi_j^2-\alpha^2-(2\eta_j-|\beta|)^2 + i 4\xi_j (2\eta_j-|\beta|)}\Bigr)+ \pi \\
  &\quad- \sideset{}{'}\sum_{k=1}^{N}\arg\Bigl(
  \frac{[(\xi_j-\xi_k)+i(\eta_j-\eta_k)][(\xi_j+\xi_k)+i(\eta_j-\eta_k)]}{
  [(\xi_j+\xi_k)+i(\eta_j+\eta_k)][(\xi_j-\xi_k)+i(\eta_j+\eta_k)]}\Bigr),
\intertext{whereas the product of a pair of weights \(C_j\) and \(\hat{C}_j = C_{N+j}\) is}
  C_j \hat{C}_j^\ast &=-4\lz_j^2 \frac{(2\lz_j+i|\beta|)^2-\alpha^2}{(2\lz_j-i|\beta|)^2-\alpha^2}
  \frac{(2\eta_j)^2}{(2\xi_j)^2} \Bigl[\sideset{}{'}\prod_{k=1}^{N}  \frac{(\lz_j-\lz_k^\ast)(\lz_j+\lz_k)}{
  (\lz_j-\lz_k)(\lz_j+\lz_k^\ast)}\Bigr]^2,
\end{align*}
where the prime indicates that the term with \(k=j\) is omitted from the sum and product.
\end{remark}
Incidentally by the argumentation of the proof of Proposition \ref{p:Nmirsol}, it follows that
the boundary matrix \(K_N(t,x,\lz)\) corresponding to the dressed solution \(\widehat{u}[N]\)
has kernel vectors \(\js_0' = D[2N](t,x,\lz_0)\js_0\) and the orthogonal vector as before
\(\jss_0' = D[2N](t,x,\lz_0^\ast)\jss_0\) respectively at the spectral parameter \(\lz=\lz_0\)
and \(\lz_0^\ast\) as described in the first case. Additionally, note that the second case,
in which the solution \(\js_0\) and the kernel vector \(\upsilon\) are linearly independent,
can also occur. An example is given by the non-zero seed solution \(u[0] = \rho e^{2i\rho^2 t}\)
with constant background \(\rho > 0\) in the case of the Robin boundary condition, where the
solutions to the Lax system can not be connected to the kernel vectors.
\subsection{Soliton reflection}\label{sc:sol}
The Darboux transformation presented in Appendix \ref{app:dar} gives the algebraic means to derive
\(N\)-soliton solutions simply by calculating the \((12)\)-entry of the projector matrices \((P[j])_{12}\)
for \(j=1,\dots,N\) recursively and then sum them up or by the direct calculation of the quotient of two
\(2N\times2N\) matrices, which represents the \((12)\)-entry of the sum of projector matrices, i.e.\@
\((\Sigma_1)_{12}\), as presented in \cite{Zh}. Motivated by Section \ref{sc:rel}, the pure soliton
solutions in the case of the new boundary condition, which we can obtain, are constructed by choosing
pairs of spectral parameter \(\lz_j\) and \(\widehat{\lz}_j^\ast\), \(j=1,\dots,N\), and associated
constants \(u_j\), \(v_j\), \(-\hat{v}^\ast_j\) and \(\hat{u}^\ast_j\) as explained therein.

For \(N=1\), consider the spectral parameter \(\lz_1 = \xi_1 + i \eta_1\), where it is comprehensible
that, with regard to \eqref{eq:1sol}, \(\xi_1\) and \(\eta_1\) respectively describe the velocity and
the amplitude of the physical one-soliton. Further, the quotient of the constants \(u_1\) and \(v_1\)
is highly related to the initial position \(x_1\) and phase \(\phi_1\) of the soliton. Consequently,
the mirror soliton corresponding to \(\widehat{\lz}_1^\ast=-\xi_1+i\eta_1\) has opposite velocity to
and the same amplitude as the physical soliton. Particularly, we have visualized said behavior in
Figures \ref{f:dir2} and \ref{f:new2}. Whereas, the Dirichlet boundary condition \(u(t,0)=0\) occur as
a special case of the new boundary condition \eqref{eq:nbc}, when for example \(|\alpha|\to\infty\),
\(|\beta|\to\infty\) or \(\beta\to0\). Indeed, structurally these cases correspond to the boundary matrix
\(K_0(t,x,\lz)=\I\). Thereby, we plotted in Figure \ref{f:dir2} on the left the reflection of a one-soliton
solution \(|\widehat{u}[1](t,x)|\) subject to the Dirichlet boundary condition as well as on the right a
contour plot, which includes the mirror soliton (dashed). In Figure \ref{f:new2}, we chose particular
parameter \(\alpha=1\) and \(\beta=2\) to plot an example of a one-soliton solution \(|\widehat{u}[1](t,x)|\)
in the case of soliton reflection with respect to the new boundary condition in three dimensions on the left
and as a contour plot together with the mirror soliton on the right. It is observable that in these cases
the physical soliton and the mirror soliton change roles after the usual soliton interaction with the
physical soliton visible before and the mirror soliton visible after the interaction with the boundary.
Additionally, in the case of the Dirichlet boundary condition the interaction of the pair of solitons
is such that the whole solution is zero at the boundary \(x=0\).
\begin{figure}
  \centering
  \includegraphics[width=\textwidth]{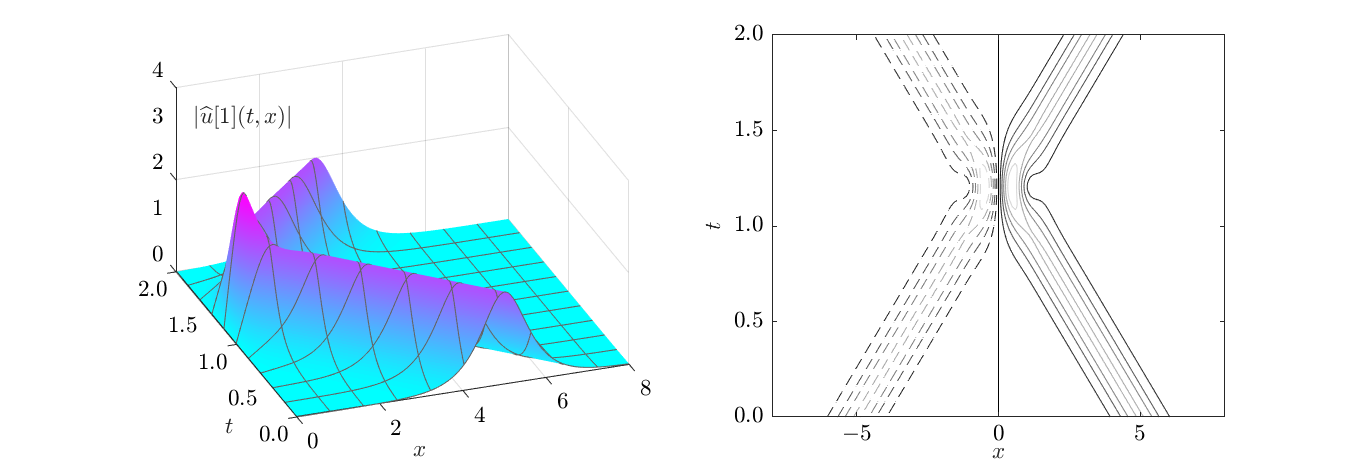}
  \caption{\small Dirichlet boundary condition at \(x=0\): one-soliton reflection with \(\xi_1 =1\), \(\eta_1 =1\),
  \(x_1 = 5\) and \(\phi_1=0\). Left: 3D plot of \(|\widehat{u}[1](t,x)|\). Right: contour plot showing the mirror
  soliton (dashed) to the left of \(x=0\).}\label{f:dir2}
\end{figure}
\begin{figure}
  \centering
  \includegraphics[width=\textwidth]{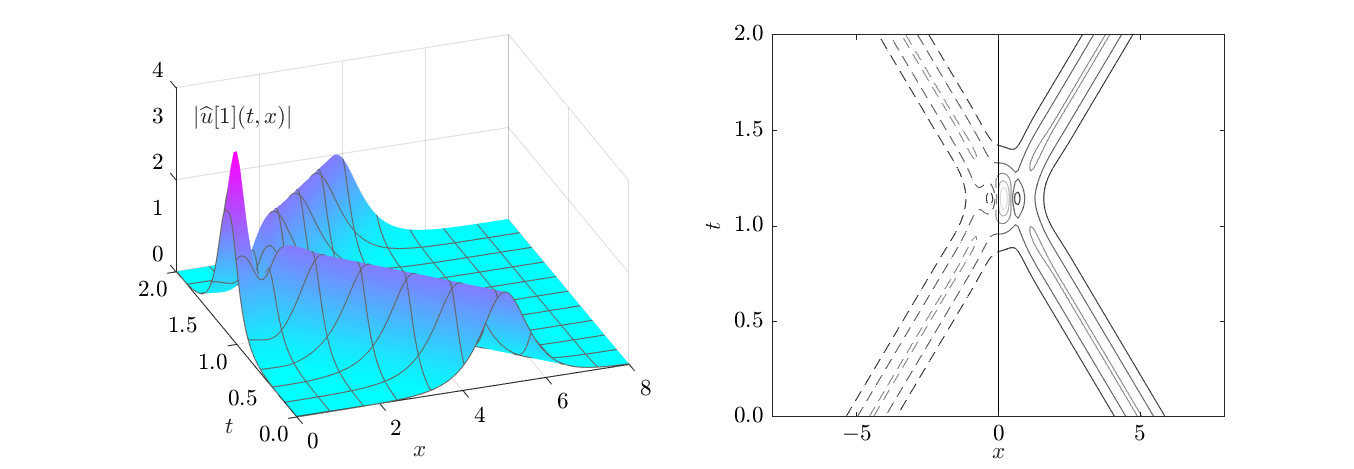}
  \caption{\small New boundary condition (\(\alpha=1\), \(\beta=2\)) at \(x=0\): one-soliton reflection with
  \(\xi_1 =1\), \(\eta_1 =1\), \(x_1 = 5\) and \(\phi_1=0\). Left: 3D plot. Right: contour plot.}\label{f:new2}
\end{figure}
\begin{figure}
  \centering
  \includegraphics[width=\textwidth]{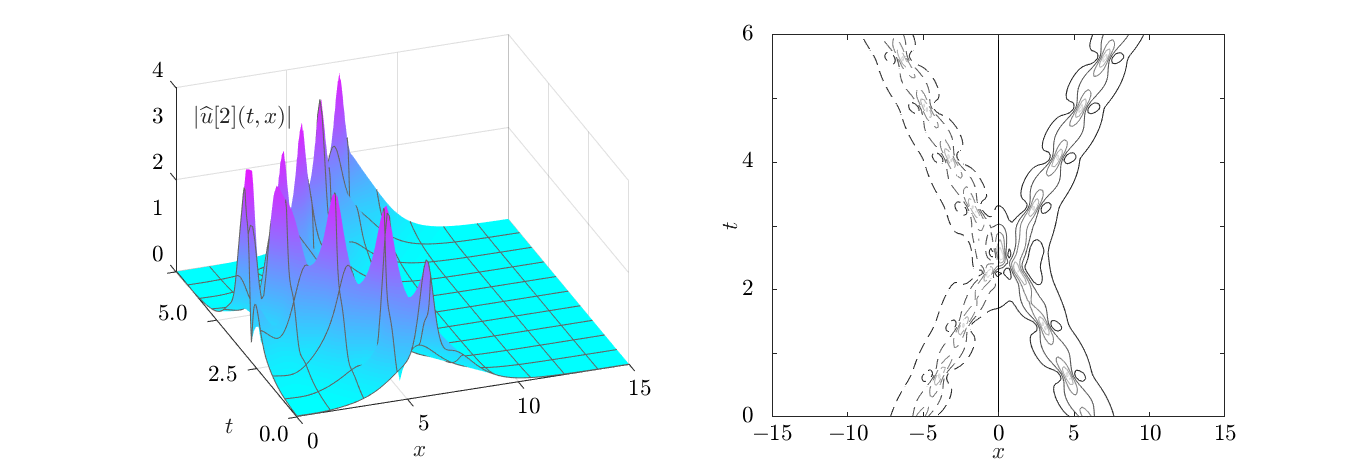}
  \caption{\small New boundary condition (\(\alpha=1\), \(\beta=2\)) at \(x=0\): two-soliton reflection with
  \(\xi_1 =\xi_2 =1/2\), \(\eta_1 =1/2\), \(\eta_2 =3/2\), \(x_1 =x_2 = 5\) and \(\phi_1=\phi_2=0\).
  Left: 3D plot. Right: contour plot.}\label{f:new4}
\end{figure}

Subsequently, we used the mentioned algorithm to include higher order soliton solutions in the results.
First of all, inspired by the breather in the case of the Dirichlet boundary condition, see Figure 6 in \cite{BiHw},
we plotted a similar breather solution as soliton reflection in the case of the new boundary condition
with parameter \(\alpha=1\) and \(\beta=2\) on the left and the contour together with the mirror soliton
on the right of Figure \ref{f:new4}. As one would suspect, the main difference can be observed
at the boundary \(x=0\). Ultimately, we went one step further and even plotted the reflection of a three-soliton
solution in the case of the new boundary condition, again with the same parameters, on the left and
its contour including the mirror soliton on the right of Figure \ref{f:new6}. The choice of parameters,
which is needed in order to comply with the conditions of Proposition \ref{p:Nmirsol}, is described
in Section \ref{sc:rel}.
\begin{figure}
  \centering
  \includegraphics[width=\textwidth]{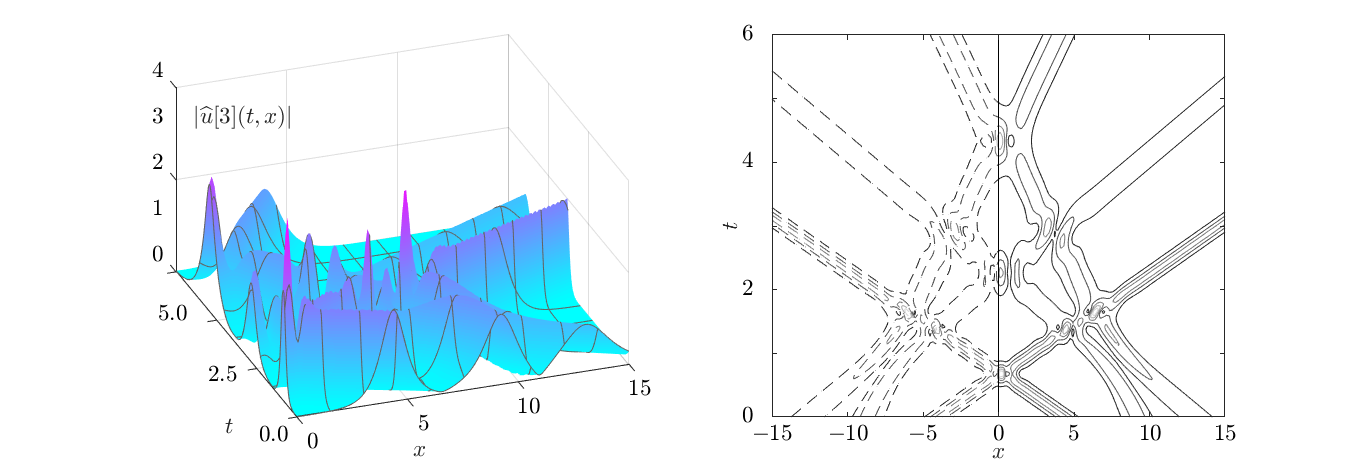}
  \caption{\small New boundary condition (\(\alpha=1\), \(\beta=2\)) at \(x=0\): three-soliton reflection with
  \(\xi_1 =3/2\), \(\xi_2 =1/2\), \(\xi_3 =5/4\), \(\eta_1 =1\), \(\eta_2 =3/4\), \(\eta_3 =1/2\), \(x_1 = 5\), \(x_2 = 8\),
  \(x_3 = 11\) and \(\phi_1=\phi_2=\phi_3=0\).
  Left: 3D plot. Right: contour plot.}\label{f:new6}
\end{figure}

\section*{Conclusion}
In this work, we further developed the method of dressing the boundary to be applicable
to the NLS equation on the half-line with the new boundary condition. The boundary condition
corresponds to a time dependent gauge transformation \eqref{eq:K0} and this time dependence
together with the polynomial degree with respect to the spectral parameter of the transformation
thin out the solution space for the new boundary condition. As we have seen in \cite{G},
for the time dependence we need the solution to go to zero as the time goes to infinity.
Moreover, the polynomial degree disables the consideration of boundary-bound solitons.
Nonetheless, we were able to show that it is possible to construct reflected pure
soliton solutions of arbitrary (even) order in this model and to visualize the result.
\newpage \begin{appendices}
\section{Darboux transformation}\label{app:dar}
The Darboux transformation can be viewed as gauge transformation acting on forms of the Lax pair
\(U\), \(V\). For that, the undressed Lax system~\eqref{eq:jost} will be written as \(U[0]\), \(V[0]\)
and \(\js[0]\) and the transformed, structural identical system as \(U[N]\), \(V[N]\) and \(\js[N]\)
with \(N\in \mathbb{N}\).
The transformed vector \(\js[N] = D[N] \js[0]\) satisfies the transformed system
\begin{equation}\label{eq:tls}
  \js[N]_x = U[N]\js[N],\quad
  \js[N]_t = V[N]\js[N],
\end{equation}
whereas they are connected by
\begin{equation}\label{eq:DN}
  D[N]_x= U[N]D[N]-D[N]U[0], \quad D[N]_t= V[N]D[N]-D[N]V[0].
\end{equation}
For given \(N\) distinct column vector solutions \(\js_j = (\mu_j,\nu_j)^\intercal\)
of the undressed Lax system \eqref{eq:jost} evaluated at \(\lz=\lz_j\), \(j=1\dots N\),
we construct an iteration of the one-fold dressing matrix \(D[1]\) in the following sense
\begin{equation*}
D[N] = ((\lambda-\lambda^\ast_N)\I+(\lambda^\ast_N-\lambda_N) P[N]) \cdots
((\lambda-\lambda^\ast_1)\I+(\lambda^\ast_1-\lambda_1) P[1]),
\end{equation*}
where \(P[j]\) are projector matrices defined by
\begin{equation}\label{eq:Pj}
P[j] = \frac{\josts_j[j-1] \josts^\dagger_j[j-1]}{\determz{\josts}{j}{[j-1]}},
\quad \josts_j[j-1]=D[j-1]\big\rvert_{\lambda=\lambda_j} \josts_j.
\end{equation}
Then to reconstruct the solution \(u[N]\), we need to insert \(\js[N] = D[N] \js[0]\) into the transformed
Lax system \eqref{eq:tls} and extract the information of the coefficient of \(\lambda^{N-1}\) of the first line.
Therefore, we need the coefficient of \(\lambda^{N-1}\) of \(D[N]\) which we denote by \(\Sigma_1\), i.e.
\begin{equation}\label{eq:S1}
  \Sigma_1=\sum_{j=1}^{N} -\lambda_j^\ast \I + (\lambda_j^\ast - \lambda_j) P[j].
\end{equation}
Consequently, the reconstruction formula can be computed as
\begin{equation*}
Q[N]=Q[0] - i \sum_{j=1}^{N} (\lambda_j-\lambda^\ast_j)[\st,P[j]].
\end{equation*}
This calculation can be used to recursively construct \(N\)-soliton solutions. Especially,
we use it to compute the solutions in Section \ref{sc:sol}. Moreover, the change of the scattering
data under Darboux transformations has been investigated, among others, for the NLS equation,
see \cite{GHZ}. With that book serving as a basis, we give a brief overlook for the relevant
theorem in our notation.
\subsection{Change of scattering data under Darboux transformations}
With scattering data \((\rho,\{\lz_j,C_j\}_{j=1}^N)\), \(\lz_j\in \C_+\) for all \(j=1,\dots,N\),
we want to give the relevant information needed to retrace the change of scattering data under
Darboux transformations. It is of renewed importance that the solution and its derivative with respect
to \(x\) connected to the scattering data is a sufficiently fast decaying function for \(|x|\to \infty\).
Then, given a spectral parameter \(\lz_0\in \C_+\setminus\{\lz_1,\dots,\lz_N\}\) and a column solution of the undressed Lax system
\begin{align*}
  \js_0 &= u_0\js_-^{(1)}(t,x,\lz_0)+v_0\js_+^{(2)}(t,x,\lz_0)\\
  &=u_0\mjs_-^{(1)}(t,x,\lz_0)e^{-i\theta(t,x,\lz_0)}+v_0\mjs_+^{(2)}(t,x,\lz_0)e^{i\theta(t,x,\lz_0)}.
\end{align*}
Defining the ratio of the second and the first component to be
\begin{equation*}
  q = \frac{[\mjs_-]_{21}(t,x,\lz_0)+\frac{v_0}{u_0}[\mjs_+]_{22}(t,x,\lz_0)e^{2i\theta(t,x,\lz_0)}}{
  [\mjs_-]_{11}(t,x,\lz_0)+\frac{v_0}{u_0}[\mjs_+]_{12}(t,x,\lz_0)e^{2i\theta(t,x,\lz_0)}},
\end{equation*}
we obtain, in turn, an expression for the ratio of \(v_0\) and \(u_0\), i.e.\@
\begin{equation}\label{eq:quo}
  -\frac{v_0}{u_0} = \frac{[\mjs_-]_{21}(t,x,\lz_0)-q[\mjs_-]_{11}(t,x,\lz_0)}{
  [\mjs_+]_{22}(t,x,\lz_0)-q[\mjs_+]_{12}(t,x,\lz_0)}e^{-2i\theta(t,x,\lz_0)}.
\end{equation}
Also, the one-fold Darboux transformation corresponding to \(\lz_0\) and \(\js_0\) takes the form
\begin{equation*}
  D[1] = \lz \I + \frac{1}{1+|q|^2}
  \begin{pmatrix}
    -\lz_0-\lz_0^\ast |q|^2 & (\lz_0^\ast-\lz_0) q^\ast \\
    (\lz_0^\ast-\lz_0) q & -\lz_0^\ast-\lz_0 |q|^2
  \end{pmatrix}.
\end{equation*}
The properties of the Jost functions imply
\begin{equation*}
  \lim_{x\to-\infty}q = \infty, \quad \lim_{x\to+\infty}q = 0.
\end{equation*}
Thereby, adding a pole to the scattering data under Darboux transformations can be explained by the following
\begin{theorem}\label{t:scadt}
Let the scattering data \(a_{11}(\lz)\), \(\lz\in \C_+\cup \R\), \(a_{12}(\lz)\), \(\lz\in \R\)
and \(b_j\) for \(j=1,\dots,N\) be given. Applying the Darboux transformation with \(\lz_0\in
\C_+\setminus\{\lz_1,\dots,\lz_N\}\) and \(\js_0 = u_0\js_-^{(1)}(t,x,\lz_0)+v_0\js_+^{(2)}(t,x,\lz_0)\),
where \(u_0\), \(v_0\in \C\setminus\{0\}\), we add an eigenvalue to the scattering
data leaving the original eigenvalues unchanged. Further,
\begin{center}
\begin{tabular}{ c c p{0.5cm} c c}
$\begin{aligned}[t]
   a'_{11}(\lz) &= \frac{\lz-\lz_0}{\lz-\lz_0^\ast} a_{11}(\lz),\\
   a'_{12}(\lz) &= a_{12}(\lz),\\
   b'_j &= \vphantom{\frac{\lz_j-\lz_0^\ast}{\lz_j-\lz_0}}b_j,\\
   b'_0 &= -\frac{v_0}{u_0}\vphantom{\frac{\lz_0-\lz_0^\ast}{a_{11}(\lz_0)}},
\end{aligned}$&
$\begin{aligned}[t]
   &\vphantom{\frac{\lz-\lz_0}{\lz-\lz_0^\ast}}\lz \in \C_+ \cup \R,\\
   &\vphantom{a'_{12}}\lz \in \R,\\
   &\vphantom{\frac{\lz_j-\lz_0^\ast}{\lz_j-\lz_0}}j=1,\dots,N,
\end{aligned}$&&
$\begin{aligned}[t]
   \rho'(\lz) &=  \frac{\lz-\lz_0^\ast}{\lz-\lz_0} \rho(\lz),\vphantom{\frac{\lz-\lz_0}{\lz-\lz_0^\ast}}\\
   \vphantom{a'_{12}}\\
   C'_j &= \frac{\lz_j-\lz_0^\ast}{\lz_j-\lz_0}C_j,\\
   C'_0 &= -\frac{v_0}{u_0} \frac{\lz_0-\lz_0^\ast}{a_{11}(\lz_0)}.
\end{aligned}$&
$\begin{aligned}[t]
   &\vphantom{\frac{\lz-\lz_0}{\lz-\lz_0^\ast}}\lz \in \R,\\
   \vphantom{a'_{12}}\\
   &\vphantom{\frac{\lz_j-\lz_0^\ast}{\lz_j-\lz_0}}j=1,\dots,N,
\end{aligned}$
\end{tabular}
\end{center}

\end{theorem}
\begin{proof}
The scattering data rely heavily on the Jost functions. That is why, the first step is to find the behavior
of the Jost functions in the transformed system. Therefore, we need to see what the limit values of the Darboux
transformation are
\begin{align*}
  \lim_{x\to-\infty}D[1] = \diag(\lz-\lz_0^\ast,\lz-\lz_0),\quad
  \lim_{x\to+\infty}D[1] = \diag(\lz-\lz_0,\lz-\lz_0^\ast).
\end{align*}
Then, we can deduce that the transformed Jost functions can be expressed through
\begin{align*}
  (\js_-^{(1)})'(t,x,\lz) = \frac{D[1]}{\lz-\lz_0^\ast} \js_-^{(1)}(t,x,\lz),\quad
  (\js_+^{(2)})'(t,x,\lz) = \frac{D[1]}{\lz-\lz_0^\ast} \js_+^{(2)}(t,x,\lz),
\end{align*}
which also is passed onto \((\mjs_-^{(1)})'\) and \((\mjs_+^{(2)})'\). As already mentioned in
Section \ref{sc:NLS}, \(a_{11}(\lz) =  \det[\js_-^{(1)}|\js_+^{(2)}]\). It follows that for \(\lz \in \C_+\cup\R\), the
limit values of \([\mjs_-]_{11}\) and \([\mjs_+]_{22}\) are \(a_{11}(\lz)\) as \(x\) goes to
\(+\infty\) and \(-\infty\), respectively. So that we have
\begin{equation*}
a'_{11}(\lz) = \lim_{x\to \infty} ([\mjs_-]_{11})' = \frac{\lz-\lz_0}{\lz-\lz_0^\ast} a_{11}(\lz).
\end{equation*}
Analogously, we find for \(a_{12}(\lz)\) that
\begin{equation*}
  a_{12}(\lz) = [\js_+]_{22}[\js_-]_{12}-[\js_+]_{12}[\js_-]_{22}
  =([\mjs_+]_{22}[\mjs_-]_{12}-[\mjs_+]_{12}[\mjs_-]_{22})e^{2i\theta(t,x,\lz)},
\end{equation*}
and therefore the limit values of \([\mjs_-]_{12}\) and \(-[\mjs_+]_{12}\) behave as
\(a_{12}(\lz)e^{-2i\theta(t,x,\lz)}\) as \(x\) goes to \(+\infty\) and \(-\infty\), respectively.
Consequently,
\begin{equation*}
a'_{12}(\lz) = \lim_{x\to \infty} ([\mjs_-]_{12})' =  a_{12}(\lz).
\end{equation*}
Also resulting in \(\rho'(\lz) =  \frac{\lz-\lz_0^\ast}{\lz-\lz_0} \rho(\lz)\). Since the Jost functions we relate
in order to obtain \(b_j\) are changed identically with \(D[1]/(\lz-\lz_0^\ast)\), \(b_j\) remain unchanged, i.e.\@
\(b'_j = b_j\) for \(j=1,\dots,N\). Then, by the definition of \(C_j\) we can calculate
\begin{equation*}
  C'_j = b'_j\Bigl(\frac{\id a'_{11}(\lz_j)}{\id \lz}\Bigr)^{-1} = \frac{\lz_j-\lz_0^\ast}{\lz_j-\lz_0} C_j,
  \quad j=1,\dots,N.
\end{equation*}
At the new eigenvalue \(\lz= \lz_0\), we have that the transformed Jost function are also
identically changed by
\begin{equation*}
\frac{D[1](t,x,\lz_0)}{(\lz_0-\lz_0^\ast)}= \frac{1}{1+|q|^2}
\begin{pmatrix}
  |q|^2 & -q^\ast \\
  -q & 1
\end{pmatrix}.
\end{equation*}
Hence, as we calculated already in \eqref{eq:quo}, we obtain
\begin{equation*}
  b'_0 = \frac{([\js_-]_{21})'(t,x,\lz_0)}{([\js_+]_{22})'(t,x,\lz_0)}
  =\frac{[\js_-]_{21}(t,x,\lz_0)-q[\js_-]_{11}(t,x,\lz_0)}{[\js_+]_{22}(t,x,\lz_0)-q[\js_+]_{12}(t,x,\lz_0)}
  =-\frac{v_0}{u_0}.
\end{equation*}
Subsequently, the weight for the added eigenvalue is readily obtained by
\begin{equation*}
  C'_0 = b'_0\Bigl(\frac{\id a'_{11}(\lz_0)}{\id \lz}\Bigr)^{-1} = -\frac{v_0}{u_0} \frac{\lz_0-\lz_0^\ast}{a_{11}(\lz_0)}.
\end{equation*}
\end{proof}
\begin{remark}\label{r:scadt}
A particular example is dressing a pure soliton solution from the zero seed solution
for which \(a_{11}(\lz) = 1\), \(a_{12}(\lz) = 0\), whereby \(\rho(\lz) = 0\). Then,
inserting poles \(\lz_1,\dots,\lz_N \in \C_+\) with corresponding \(u_j,v_j\in
\C\setminus\{0\}\), \(j=1,\dots,N\), results in the (relevant) scattering data
\begin{align*}
  a_{11}^{(N)}(\lz) & = \prod_{j=1}^{N} \frac{\lz-\lz_j}{\lz-\lz_j^\ast},\quad
  a_{12}^{(N)}(\lz) = 0,\quad
  C_j^{(N)} &= -\frac{v_j}{u_j} \prod_{k=1}^{N}(\lz_j-\lz_k^\ast)\Bigl(\sideset{}{'}\prod_{k=1}^{N}(\lz_j-\lz_k)\Bigr)^{-1},
\end{align*}
where the prime indicates that the term with \(k = j\) is omitted from the product.
\end{remark}
\section{Robin boundary condition}\label{app:rob}
As mentioned in Section \ref{sc:NLS}, the proof, which we tailored to fit the
new boundary condition, is also applicable to the Robin boundary condition.
In fact, we will proceed and write it down explicitly not only for the
convenience of the interested reader but also since it was a guiding step
between the proof for the defect conditions connecting two half-lines, see \cite{G},
and Proposition \ref{p:Nmirsol}.
\subsection{Dressing the Robin boundary condition}
As before, we look at the NLS equation \eqref{eq:NLS} for \((t,x)\in
\R_+\times\R_+\) and complement it with Robin boundary condition at \(x=0\),
which, in our notation, is
\begin{equation}\label{eq:rbc}
  u_x = \alpha u, \quad \alpha\in\R.
\end{equation}
Then, the NLS equation has again a corresponding Lax system and the boundary condition can be
written in the form of a boundary constraint
\begin{align}\label{eq:rbcv}
  0&=(V(t,x,-\lz) G_0(\lz) - G_0(\lz) V(t,x,\lz))\big\rvert_{x=0},
\end{align}
where the boundary matrix \(G_0(\lz)\) is given by
\begin{equation}\label{eq:G0}
  G_0(\lz)= \frac{1}{i\alpha+2\lz}
  \begin{pmatrix}
    i\alpha-2\lz & 0\\
    0 &i\alpha+2\lz
  \end{pmatrix}
\end{equation}
and is, in particular, independent of \(t\) and \(x\). Similarly to the boundary matrix
\eqref{eq:K0}, \(G_0(\lz)\) is scaled by \((i\alpha+2\lambda)^{-1}\), so that \((G_0(\lz))^{-1}
= G_0(-\lz)\), which is crucial for the proof of Proposition \ref{p:Nmirsolrob}.
\begin{remark}\label{rem:zero}
The kernel vectors for \(G_0(\lz)\) can be easily obtained at \(\lz=\frac{i\alpha}{2}\) and
\(\lz=-\frac{i\alpha}{2}\), we have respectively \(e_1\) and \(e_2\).
\end{remark}
\subsection{Dressing the boundary}
In this approach we will also leave out boundary-bound soliton solutions.
Due to the fact that for boundary-bound soliton solutions the proof needs to be
slightly changed. The proof will be not as detailed as for Proposition \ref{p:Nmirsol}.
Nonetheless, we will point out the differences rather than the similarities.
\begin{prop}\label{p:Nmirsolrob}
Consider a solution \(u[0](t,x)\) to the NLS equation on the half-line subject
to the Robin boundary conditions \eqref{eq:rbcv} with parameter \(\alpha\in \R\).
Take one solution \(\js_0\) of the undressed Lax system corresponding to \(u[0]\)
for \(\lz=\lz_0=-\frac{i\alpha}{2}\). Further, take \(N\) solutions \(\js_j\) of
the undressed Lax system corresponding to \(u[0]\) for distinct \(\lz= \lz_j\in
\C \setminus \bigl(\R\cup i\R\bigr)\), \(j=1,\dots,N\). Constructing \(G_0(\lz)\) as
in \eqref{eq:G0} with \(\alpha\), we assume that there exist paired solutions
\(\widehat{\js}_j\) of the undressed Lax system corresponding to \(u[0]\) for
\(\lz=\widehat{\lz}_j= -\lz_j\), \(j=1,\dots,N\), satisfying
\begin{equation}\label{eq:ssG0}
\widehat{\js}_j\big\rvert_{x=0} = (G_0(\lz_j) \js_j)\big\rvert_{x=0},\quad
\widehat{\lz}_k\neq\lz_j.
\end{equation}
Then, a \(2N\)-fold Darboux transformation \(D[2N]\) using
\(\{\js_1,\widehat{\js}_1,\dots,\js_N,\widehat{\js}_N\}\)
and their respective spectral parameter lead to the solution \(u[2N]\)
to the NLS equation on the half-line. In particular, the boundary condition
is preserved and we denote such a solution \(u[2N]\) by \(\widehat{u}[N]\).
\end{prop}
\begin{proof}
Defining the matrix \(G_N(t,\lz)=2\lambda \st +G^{(0)}(t)\) similar to \(K_N(t,x,\lz)\)
through the transformed kernel vectors \(D[2N](t,x,\lz_0)e_1\) and \(D[2N](t,x,\lz_0^\ast)e_2\)
at \(x=0\) and respectively \(\lz= \lz_0\) and \(\lz= \lz_0^\ast\), we can derive that
the equality
\begin{equation}\label{eq:key}
  (D[2N](t,x,-\lz) G_0(\lz))\big\rvert_{x=0} = (G_N(t,\lz) D[2N](t,x,\lz))\big\rvert_{x=0}
\end{equation}
holds, whereas \(G_0(\lz)\) is multiplied by \(2\lambda+i\alpha\).

To reconstruct the expression of \(G_N(t,\lz)\), we analyze the equality \eqref{eq:key}.
In particular for the equality of the matrix coefficients regarding \(\lz\) of order \(2N\),
we obtain for the off-diagonal entries of \(G^{(0)}(t)\) that \(G^{(0)}_{12}(t)=0\) and \(G^{(0)}_{21}(t)=0\).
Then, for the diagonal entries, we need to evaluate the determinant of \(G_N(t,\lz)\) in two ways.
Firstly, as a product of matrices
\begin{align*}
  \det(G_N(t,\lz)) &=\det(D[2N](t,0,-\lz))\det(G_0(\lz))\det((D[2N](t,0,\lambda))^{-1})
  = -4\lambda^2 -\alpha^2.
\end{align*}
Secondly, through the definition \(G_N(t,\lz)=2\lambda \st +G^{(0)}(t)\) and the partial result for
the off-diagonal entries, we obtain consequently
\begin{equation}\label{eq:Gdet}
\begin{aligned}
  G^{(0)}_{11}(t)-G^{(0)}_{22}(t) &= 0, \qquad \det(G^{(0)}(t))&=-\alpha^2.
\end{aligned}
\end{equation}
Hence, we need to have \(G^{(0)}(t) = \pm i\alpha \I\). However, by the equality \eqref{eq:key} of the
matrix coefficients regarding \(\lz\) of zero-th order, we can verify, that the structure of \(G_0(\lz)\)
is preserved, since we need to have \(G^{(0)}(t) = - i\alpha \I\) in order for
\begin{equation}\label{eq:S2N}
  -i\alpha\Sigma_{2N}=G^{(0)}(t)\Sigma_{2N}
\end{equation}
to hold.
\end{proof}
This proposition serves as a means to insert solitons in the case of the zero seed solution and also
the non-zero seed solution with constant background as presented in \cite{Zh}. Similarly, to the
argumentation therein, we would also need to proof distinctly the dressing of boundary-bound
solitons with pure imaginary spectral parameters \(\lz_j\in i\R\), \(j=1,\dots,N\). However,
since the reasoning is similar and it is presumably not applicable in the case of the new
boundary condition, we omit that consideration here.

Nonetheless, we hereby gave the proof to dress solitons corresponding to \(\lz_j\in \C\setminus
\bigl(\R\cup i\R\bigr)\), \(j=1,\dots,N\), in the dressing the boundary method we adapted
to the new boundary conditions. It turns out that the equality \eqref{eq:S2N}, from which we then
follow the definite form of \(G_N(\lz)\), functions as an intermediate step between the ideas in
the proof for the defect conditions connecting two half-lines \cite{G} and for the new boundary
conditions, see Proposition \ref{p:Nmirsol}.
\subsection{Relations between scattering data}
Consider the zero seed solution \(u[0]=0\) and \(\C\setminus\R\ni \lz_0 = -\frac{i\alpha}{2}\).
Following the steps in Section \ref{sc:rel}, we take \(\js_j = e^{-i(\lz_j x+ 2 \lz_j^2 t)\st}(u_j,v_j)^\intercal\)
for \(\lz =\lz_j\) and with respect to the relation \(\widehat{\js}_j \big\rvert_{x=0} = (G_0(\lz_j) \js_j)
\big\rvert_{x=0}\), \(-\lz_k\neq\lz_j\) for all \(j,k\in\{1,\dots,N\}\), we also take \(\widehat{\js}_j =
e^{-i(\widehat{\lz}_j x + 2 \widehat{\lz}_j^2 t)\st} (\hat{u}_j,\hat{v}_j)^\intercal\) for \(\lz =
\widehat{\lz}_j = -\lz_j\), whereas
\begin{equation*}
  \frac{\hat{u}_j}{\hat{v}_j} = \frac{i\alpha-2\lz_j}{i\alpha+2\lz_j} \frac{u_j}{v_j}, \quad j=1,\dots,N.
\end{equation*}
Analogously, if we apply a two-fold Darboux transformation consisting of \(\{\lz_1, \js_1, \widehat{\lz}_1^\ast,
\widehat{\jss}_1\}\), we obtain the scattering data, which particularly results in the relation
\begin{equation}\label{eq:c1c2}
  C_1^{(2)}(C_2^{(2)})^\ast = -4 \lambda_1^2\cdot \frac{i\alpha-2\lz_1}{i\alpha+2\lz_1} \cdot
  \frac{\Im(\lz_1)^2}{\Re(\lz_1)^2},
\end{equation}
which is up to notation the same as in \cite{BiHw}. To align the notation, one would need to
complex conjugate \eqref{eq:c1c2} and then it would be equal to the equation (2.36) in their paper
with \(k_1= -\lambda_1^\ast\). This is due to the differently defined potential \(\widetilde{Q}\)
of the matrix \(V\), which as a consequence gives the existence of Jost solutions with different
asymptotic behavior and continuations into different parts of the complex plane. Analogously to
Section \ref{sc:rel}, we obtain the following relations between the initial positions and phases
of \(2N\) inserted solitons.
\begin{remark}
In general, we can construct a \(2N\)-Darboux transformation using the information given by
\(\{\lz_1, \js_1, \dots, \lz_N, \js_N, \widehat{\lz}_1^\ast, \widehat{\jss}_1,\dots,
\widehat{\lz}_N^\ast, \widehat{\jss}_N\}\), where \(\lz_j = \xi_j + i \eta_j\) and consequently
\(\widehat{\lz}_1^\ast = -\xi_j + i \eta_j\) for \(j=1,\dots,N\) with corresponding solutions to
the undressed Lax system as above. Then for \(j=1,\dots,N\), the relation for a pair of initial
positions \(x_j\) and \(\widehat{x}_j = x_{N+j}\) as well as phases \(\phi_j\) and
\(\widehat{\phi}_j = \phi_{N+j}\) amounts to
\begin{align*}
  x_j + \hat{x}_j & = \frac{1}{2\eta_j}\log\Bigl(1+\frac{\eta_j^2}{\xi_j^2}\Bigr)+
  \frac{1}{4\eta_j}\log\Bigl(\frac{(2\xi_j)^2 +(\alpha- 2\eta_j)^2}{(2\xi_j)^2
  +(\alpha+ 2\eta_j)^2}\Bigr)\\
  &\quad -\frac{1}{2\eta_j} \sideset{}{'}\sum_{k=1}^{N} \log
  \frac{[(\xi_j-\xi_k)^2+(\eta_j-\eta_k)^2][(\xi_j+\xi_k)^2+(\eta_j-\eta_k)^2]}{
  [(\xi_j+\xi_k)^2+(\eta_j+\eta_k)^2][(\xi_j-\xi_k)^2+(\eta_j+\eta_k)^2]},\\
  \varphi_j - \hat{\varphi}_j & = 2\arg(\lz_j) + \arg\Bigl(\frac{2\xi_j+i(2\eta_j-\alpha)}{
  2\xi_j+i(2\eta_j+\alpha)}\Bigr) \\
  &\quad- \sideset{}{'}\sum_{k=1}^{N}\arg\Bigl(
  \frac{[(\xi_j-\xi_k)+i(\eta_j-\eta_k)][(\xi_j+\xi_k)+i(\eta_j-\eta_k)]}{
  [(\xi_j+\xi_k)+i(\eta_j+\eta_k)][(\xi_j-\xi_k)+i(\eta_j+\eta_k)]}\Bigr),
\intertext{whereas the product of a pair of weights \(C_j\), \(\widehat{C}_j = C_{N+j}\) is}
  C_j \widehat{C}_j^\ast &=-4\lz_j^2 \frac{i\alpha-2\lz_j}{i\alpha+2\lz_j}\frac{(2\eta_j)^2}{(2\xi_j)^2}
  \Bigl[\sideset{}{'}\prod_{k=1}^{N}  \frac{(\lz_j-\lz_k^\ast)(\lz_j+\lz_k)}{
  (\lz_j-\lz_k)(\lz_j+\lz_k^\ast)}\Bigr]^2.
\end{align*}
\end{remark}
\end{appendices}


\begin{thebibliography}{1}

\bibitem{A}
M.~J. Ablowitz, B.~Prinari, and A.~D. Trubatch.
\newblock {\em Discrete and continuous nonlinear Schr{\"o}dinger systems},
  Volume 302.
\newblock Cambridge University Press, 2004.

\bibitem{BiHw}
G.~Biondini and G.~Hwang.
\newblock Solitons, boundary value problems and a nonlinear method of images.
\newblock {\em Journal of Physics A: Mathematical and Theoretical}, 42(20),
  2009.

\bibitem{Fu}
A.~S. Fokas and B.~Pelloni.
\newblock {\em Unified transform for boundary value problems: Applications and
  advances}.
\newblock SIAM, 2014.

\bibitem{G}
K.~T. Gruner.
\newblock Soliton solutions of the nonlinear {Schr{\"o}dinger} equation with
  defect conditions.
\newblock {\em arXiv preprint arXiv:1908.05101}, 2019.

\bibitem{GHZ}
C.~Gu, A.~Hu, and Z.~Zhou.
\newblock {\em Darboux transformations in integrable systems: theory and their
  applications to geometry}, volume~26.
\newblock Springer Science \& Business Media, 2006.

\bibitem{Za}
C.~Zambon.
\newblock The classical nonlinear {Schr{\"o}dinger} model with a new integrable
  boundary.
\newblock {\em Journal of High Energy Physics}, 2014(8):36, 2014.

\bibitem{Zh}
C.~Zhang.
\newblock Dressing the boundary: On soliton solutions of the nonlinear
  {Schr{\"o}dinger} equation on the half-line.
\newblock {\em Studies in Applied Mathematics}, 2018.

\bibitem{ZCZ}
C.~Zhang, Q.~Cheng, and D.-J. Zhang.
\newblock Soliton solutions of the sine-{G}ordon equation on the half line.
\newblock {\em Applied Mathematics Letters}, 86:64--69, 2018.

\end{thebibliography}

\end{document}